\documentclass[final,3p,times,fleqn,sort&compress]{elsarticle}
\usepackage[pdf]{lfmath}
\usepackage{amsthm}
\usepackage{lineno}
\journal{Physica D}
\newtheorem{proposition}{Proposition}
\newtheorem{lemma}{Lemma}

\begin{document}
\allowdisplaybreaks
\begin{frontmatter}
\title{Energy and potential enstrophy flux constraints in  quasi-geostrophic models}
\author[drlf]{Eleftherios Gkioulekas}
\ead{drlf@hushmail.com}
\address[drlf]{University of Texas-Pan American, Department of Mathematics, 1201 West University Drive, Edinburg, TX 78539-2999}

\begin{abstract}
We investigate an inequality constraining the energy and potential enstrophy flux spectra in two-layer and multi-layer quasi-geostrophic models. Its physical significance is that it can diagnose whether any given multi-layer model  that allows co-existing downscale cascades  of energy and potential enstrophy can allow the downscale energy flux to become large enough to yield a mixed energy spectrum where the dominant $k^{-3}$ scaling is overtaken by a subdominant $k^{-5/3}$  contribution beyond a transition wavenumber $k_t$ situated in the inertial range. The validity of the flux inequality implies that this scaling transition cannot occur within the inertial range, whereas a violation of the flux inequality beyond some wavenumber $k_t$  implies the existence of a scaling transition near that wavenumber. This flux inequality holds unconditionally in two-dimensional Navier-Stokes turbulence, however, it is far from obvious that it continues to hold in multi-layer quasi-geostrophic models, because the dissipation rate spectra for energy and potential enstrophy no longer relate in a trivial way, as  in two-dimensional Navier-Stokes. We derive the general form of the energy and potential enstrophy dissipation rate spectra for a generalized symmetrically coupled multi-layer model. From this result, we prove that in a symmetrically coupled multi-layer quasi-geostrophic model, where the dissipation terms for each layer consist of the same Fourier-diagonal linear operator  applied on  the streamfunction field  of only the same layer, the flux inequality continues to hold. It follows that a necessary condition to violate the flux inequality is the use of asymmetric dissipation where different operators are used on different layers. We explore dissipation asymmetry further in the context of a two-layer quasi-geostrophic model and derive upper bounds on the asymmetry that will allow the flux inequality to  continue to hold. Asymmetry is introduced both via an extrapolated Ekman term, based on a 1980 model by Salmon, and via differential small-scale dissipation. The results given are mathematically rigorous and require no phenomenological assumptions about the inertial range. Sufficient conditions for violating the flux inequality, on the other hand, require phenomenological hypotheses, and will be explored in future work.
\end{abstract}

\begin{keyword}
two-dimensional turbulence \sep
quasi-geostrophic turbulence \sep
two-layer quasi-geostrophic model \sep
flux inequality
\end{keyword}
\end{frontmatter}


\section{Introduction}

It is now well-known that in two-dimensional Navier-Stokes  turbulence,  most of the energy tends to go towards larges scales and most of the enstrophy tends to go towards small scales, sometimes forming an upscale inverse energy cascade with energy spectrum scaling as $k^{-5/3}$ and a downscale enstrophy cascade with $k^{-3}$ scaling \cite{article:Kraichnan:1967:1,article:Leith:1968,article:Batchelor:1969}, where $k$ is the wavenumber. Kraichnan \cite{article:Kraichnan:1967:1} argued, differently from \Fjortoft\ \cite{article:Fjortoft:1953}, that the direction of the two cascades can be justified via a thermodynamic argument in which we introduce, without proof, the assumption that the energy and enstrophy fluxes should tend to revert the energy spectrum from a cascade configuration to the absolute equilibrium configuration. The existence of forcing and dissipation arrests this tendency, thus keeping the system locked in a steady-state forced-dissipative configuration away from absolute equilibrium.

Less well-known is the fact that there is a serious error with the original \Fjortoft argument: \Fjortoft  claimed that the twin detailed conservation laws of energy and enstrophy alone imply that in every triad interaction group, more energy is transferred upscale than downscale. However, a more rigorous analysis shows that there exist triad interaction groups in which more energy is sent downscale than upscale, and it is not obvious, without additional considerations,  which group is dominant \cite{article:Warn:1975,article:Tung:2006}. Aside from this matter, the fundamental problem that underlies every other proof that utilizes only the twin conservation laws of enstrophy and energy, is that an additional assumption needs to be introduced to overcome the symmetry of the Euler equations under time reversal.  Typical assumptions, such as the tendency of the energy spectrum to revert to absolute equilibrium, or the tendency of an energy peak to spread, typify ad hoc constraints imposed implicitly on the initial conditions that are needed to break the time reversal symmetry \cite{article:Tung:2007:1}. In Ref.~\cite{article:Tung:2007:1} we counterproposed  a very simple and mathematically rigorous proof that avoids the need for any ad hoc assumptions  by considering the combined effect of the Navier-Stokes nonlinearity and the dissipation terms. The only assumption used by this proof is that the forcing spectrum is restricted to a finite interval $[k_1, k_2]$ of wavenumbers, however even that assumption can be relaxed to some extent, although not entirely eliminated \cite{thesis:Farazmand:2010,article:Protas:2011}. 

The essence of the argument in Ref.~\cite{article:Tung:2007:1} is to show that for every wavenumber $k$ not in the forcing range, the energy flux $\Pi_E (k)$ and the enstrophy flux $\Pi_G (k)$ satisfy the inequality $k^2\Pi_E (k)- \Pi_G(k) \leq 0$. Here, $\Pi_E (k)$ represents the amount of energy per unit volume transferred from the wavenumbers in the $(0, k)$ interval to the wavenumbers in the $(k,+\infty)$ interval, and $\Pi_G (k)$ is defined similarly for the enstrophy. From this inequality we then derive the following integral constraints for $\Pi_E (k)$ and $\Pi_G (k)$:
\begin{align}
\int_0^k q &\Pi_E (q) \;\df{q} \leq 0,\;\forall k\in (k_2, +\infty),\\
\int_k^{+\infty} & q^{-3} \Pi_G (q) \; \df{q} \geq 0,\;\forall k\in (0,k_1).
\end{align}
These constraints imply a predominantly upscale transfer of energy and a predominantly downscale transfer of enstrophy. The original flux inequality $k^2\Pi_E (k)- \Pi_G(k) < 0$ itself can also be directly interpreted as a tight constraint on the downscale energy flux.

The flux inequality is directly relevant to the cascade superposition hypothesis that was initially proposed in the context of two-dimensional Navier-Stokes turbulence \cite{article:Tung:2005,article:Tung:2005:1}, according to which,  for the case of finite small-scale dissipation viscosity, the downscale enstrophy cascade is accompanied with a hidden downscale energy cascade, associated with an accompanying  small downscale energy flux. We stress that the existence of this small downscale energy flux is not in doubt. The two distinct hypotheses are that: (a) the downscale energy flux is part of a downscale energy cascade that coexists with the downscale enstrophy cascade; (b) given two coexisting cascades of energy and enstrophy the corresponding energy spectra and structure functions will combine linearly. The first hypothesis can be accounted for by the argument given in section 3.2 of Ref.~\cite{article:Tung:2005}, where it is shown, leveraging an old argument by Kraichnan \cite{article:Kraichnan:1967:1}, that triad interactions with scaling exponent $-3$ transfer energy without transferring enstrophy and triad interactions with scaling exponent $-5$ transfer enstrophy without transferring energy. Consequently, there is nothing in the Navier-Stokes nonlinearity to prevent a linear superposition of two sets of triad interactions, one with scaling exponent $-3$ and one with scaling exponent $-5$, which would give rise to coexisting constant fluxes of energy and enstrophy and presumably coexisting cascades. The second hypothesis follows from the linear structure of the statistical theory of randomly forced Navier-Stokes equations; this linearity is lost by most attempts at closure modeling. These hypotheses are controversial because coexisting cascades have not been observed in the two-dimensional turbulence energy spectrum. One the other hand, they have recently been observed in models of stratified turbulence \cite{article:Kurien:2011,article:Lindborg:2011}, and   Tung and Orlando  \cite{article:Orlando:2003} have   provided evidence that they can also be observed in two-layer quasi-geostrophic turbulence, which is just one step away from two-dimensional turbulence.

The key insight to take from Refs. \cite{article:Tung:2005,article:Tung:2005:1} is that the observability of the subdominant downscale energy cascade is not decided by the nonlinearity, the conservation laws, or the relationship between the energy and enstrophy spectra alone. The correct mechanism is that the nonlinearity, combined with the configuration of the dissipation terms, determine the relation between the dissipation rate spectra of energy and enstrophy. This relation determines  whether or not the aforementioned flux inequality is satisfied, which in turn decides the observability of the subdominant downscale energy cascade. If the downscale energy flux associated with the $k^{-5/3}$ term is strong enough, then a scaling transition in the energy spectrum from $k^{-3}$ to $k^{-5/3}$ should occur near a transition wavenumber $k_t \approx \sqrt{\gn_{uv}/\gee_{uv}}$, with $\gn_{uv}$ the downscale enstrophy flux and $\gee_{uv}$ the downscale energy flux. The validity of the flux inequality for all wavenumbers $k$ in the downscale inertial range of two-dimensional Navier-Stokes turbulence implies that the downscale energy flux $\gee_{uv}$ is too weak to cause an observable  scaling transition anywhere within the inertial range. On the other hand, it is far from obvious that the flux inequality will remain unconditionally valid in quasi-geostrophic models. A violation of the flux inequality beyond some wavenumber $k_t$ in quasi-geostrophic models would imply the occurrence of a scaling transition near that wavenumber.

The goal of the present paper is to extend the flux inequality to quasi-geostrophic models. We will specifically focus on vertical discretizations of the quasi-geostrophic model, namely the $n$-layer model, and the special case of the two-layer model, with all layers having the same thickness, in terms of pressure coordinates,  on both models. From a physical standpoint, both models sacrifice the surface quasi-geostrophic dynamics at the bottom boundary, but they are otherwise good models of atmospheric turbulence for scales down to an estimated length scale of $100$km  \cite{article:Lindborg:2007}. I should like to emphasize from the beginning that in spite of any mathematical or phenomenological similarities, extending the flux inequality to quasi-geostrophic models is neither obvious nor straightforward. An overlooked fundamental difference between  two-dimensional Navier-Stokes turbulence and quasi-geostrophic turbulence is that there are many more possible configurations for the dissipation terms in quasi-geostrophic models than there are in two-dimensional Navier-Stokes. Dissipation terms are usually ignored because physical intuition alone may suggest that they  should not have an effect on the nonlinear dynamics in inertial ranges. This line of reasoning ignores that the actual configuration of the dissipation terms can still have unexpected effects on the \emph{magnitude} of the energy and potential enstrophy fluxes passing through the inertial range. These flux  effects are the underlying matter of interest motivating  the investigation initiated by the present paper. 

The original motivation underlying the aforementioned numerical investigation \cite{article:Orlando:2003} of the two-layer quasi-geostrophic model  was to show that it can reproduce the Nastrom-Gage energy spectrum of the atmosphere \cite{article:Gage:1979,article:Nastrom:1986,article:Jasperson:1984,article:Gage:1984}. However, the Nastrom-Gage controversy, reviewed to some extent in previous papers \cite{article:Tung:2006,article:Gkioulekas:p15}, is not the main concern or motivation of this paper.  Our main interest in this problem stems from the following  considerations: first, quasi-geostrophic models are simple enough that they could be accessible to investigation via   theoretical techniques  developed for two-dimensional turbulence \cite{article:Lebedev:1994,article:Lebedev:1994:1,article:Yakhot:1999,article:Procaccia:2002,article:Gkioulekas:2008:1,article:Gkioulekas:p14}. Furthermore, the possibility of being able to study a downscale energy cascade arising in the context of a two-dimensional model is particularly exciting from the point of view of the turbulence theorist, because it ties into the open question of why the downscale energy cascade of three-dimensional turbulence has intermittency corrections but the inverse energy cascade of two-dimensional turbulence does not \cite{article:Vergassola:2000,article:Tabeling:2002}.  Is it an the effect of dimension number or cascade direction? In light of such questions, an observable downscale energy cascade in a two-dimensional system is interesting in and of itself.

Mathematical results concerning the flux inequality in quasi-geostrophic models can be organized into two categories: (a) sufficient conditions for the satisfaction of the flux inequality within the entire inertial range; and (b) sufficient conditions for violating the flux inequality beyond some transition wavenumber $k_t$  within the inertial range. Results of the first type can be proved rigorously without  ad hoc phenomenological assumptions on the behavior of the energy and potential enstrophy spectra. Results of the second type require the introduction of phenomenological assumptions about the distribution of energy and potential enstrophy between layers. Consequently, the scope of this paper has been limited to what we can prove rigorously. More powerful results that can be obtained by introducing phenomenological hypotheses will be explored in future publications. 
Because the details of our argument are very technical, we will now summarize the main argument of the paper as follows.

For the generalized case of an $n$-layer model, we consider the general case of  a \emph{streamfunction dissipation} configuration, where for each   layer the dissipation terms are given by a linear differential operator applied on the streamfunction of the same layer, without entangling any streamfunctions of any other layers. The dissipation rate spectra for both energy and potential enstrophy are derived under this general configuration. Then, we specialize to the case of \emph{symmetric streamfunction dissipation}, where we assume that the corresponding dissipation operators are identical layer-by-layer. We will show that under symmetric streamfunction dissipation the flux inequality is satisfied for all wavenumbers in the inertial and dissipation range.  We note that this result is non-trivial since, beyond establishing cascade directions, it also implies  bounds on the subdominant downscale energy flux, that are tight enough to keep the underlying downscale energy cascade hidden. For the case of the two-layer quasi-geostrophic model, we consider an asymmetric configuration of dissipation terms and establish results of the form that if the asymmetry is sufficiently small, the flux inequality will remain valid. As was previously explained, we limit ourselves to results of this form because this is as far as one can go with rigorous proofs from first principles. 

From a physical standpoint, asymmetry in the dissipation between the two layers usually originates from the Ekman term, modeling the effect of friction with the surface boundary layer. However, for reasons that will be discussed more extensively at the conclusion of this paper, we will introduce an additional source of asymmetry via the small-scale dissipation terms by employing an increased viscosity or hyperviscosity coefficient at the bottom layer relative to the coefficient at the top layer. We believe that this asymmetric small-scale dissipation can facilitate a breakdown of the flux inequality, thereby allowing the downscale energy flux rate to be sufficiently strong to yield the transition to $k^{-5/3}$ scaling in the inertial range. We will see that asymmetric small-scale dissipation indeed tightens the bounds on the parameter space wherein the flux inequality is satisfied. 

Another aspect of the dissipation term configuration, that will be shown to have significant impact on the flux inequality, concerns the modeling of the Ekman term. In a typical formulation of the two-layer quasi-geostrophic model, it is usually assumed that Ekman dissipation is dependent only on the streamfunction field of the bottom potential vorticity layer. However, an alternate formulation of the two-layer quasi-geostrophic model by Salmon \cite{article:Salmon:1980}, requires that the Ekman term at the lower layer be dependent on the streamfunction  fields of both layers. To explain why, one must recall that the  two-layer model is an extreme vertical discretization of the full quasi-geostrophic model, which consists of a relative vorticity equation, a temperature equation, and additional constraining conditions. In a general multi-layer model, the relative vorticity equations are discretized in  horizontal layers that are interlaced with the discretization layers of the temperature equations. Thus, for the case of the two-layer model we have altogether $5$ physically relevant layers: the surface boundary layer corresponding approximately to $1$Atm, the lower relative vorticity layer at $0.75$Atm, the temperature midlayer at $0.5$Atm, the upper relative vorticity layer at $0.25$Atm, and the top boundary layer at $0$Atm. The potential vorticity equations are derived from the relative vorticity equations by eliminating the temperature field from the system of equations, thereby placing the potential vorticity field and the corresponding streamfunction field at the $0.25$Atm and $0.75$Atm layers. As noted by Ref. \cite{article:Salmon:1980}, the Ekman dissipation term is dependent on the streamfunction field at the surface boundary layer near $1$Atm, which can be linearly extrapolated from the streamfunction field at the lower and upper layer ($0.75$Atm and $0.25$Atm correspondingly). Consequently, even though the Ekman term is still placed on the lower-layer, owing to the linear extrapolation of the surface streamfunction field, it is dependent on the streamfunction field of both the lower and upper layers. 

It should be noted that for physical reasons, the potential vorticity layers need to remain fixed at $0.25$Atm and $0.75$Atm respectively. This corresponds to the physical assumption that the two fluid layers have equal thickness, which is a necessary assumption for atmospheric modeling \cite{book:Salmon:1998}. The surface layer driving Ekman dissipation, on the other hand, can be placed anywhere between the surface layer at $1$Atm and the lower streamfunction field layer at $0.75$Atm. When the surface layer and the lower streamfunction layer coincide, this corresponds to the usual standard Ekman term. When the two layers do \emph{not} coincide, it corresponds to the more general case of \emph{extrapolated Ekman dissipation}. For the present paper, we retain generality by parameterizing the placement of the surface boundary layer via an adjustable parameter $\mu$, and show that our main propositions are valid for the entire range of the parameter $\mu$. We will see that an increasing separation between the Ekman surface layer and the bottom potential vorticity layer   tightens the bounds on the parameter space wherein the flux inequality is satisfied. For oceanographic modeling, as well as for the purpose of satisfying basic scientific curiosity, it would be interesting to consider two-layer quasi-geostrophic models with layers having unequal thickness. Due to mathematical complications, we will  not pursue this generalization in the present paper. Nevertheless, the importance of symmetric vs. asymmetric Ekman dissipation in the context of oceanographic modelling is a relevant problem that has been investigated by a previous study \cite{article:Arbic:2007}. 

Admittedly, both Salmon's idea of extrapolated Ekman dissipation and the  idea of differential small-scale dissipation, proposed in this paper,  can be considered controversial.  On the other hand, in the context of investigating the flux inequality, it is important to be thorough about considering every interesting configuration of the dissipation terms, to determine how much impact various choices of dissipation term configurations have on the robustness of the flux inequality. Furthermore, as will become apparent from the  results of this paper, the dissipation configurations explored here  are good candidates for a dissipation filter that could violate the flux inequality and ensure a controlled downscale energy dissipation rate in numerical simulations that exceeds the restrictions that are typical in two-dimensional turbulence. 

The paper is organized as follows. In section 2 we give the governing equations for the generalized multi-layer model and discuss its conservation laws, the definition of the energy spectrum $E(k)$, potential enstrophy spectrum $G(k)$, and their relationship via the streamfunction spectrum $C_{\ab}(k)$. In section 3, after a brief recapitulation of the flux inequality for the simple case of two-dimensional Navier-Stokes turbulence, we establish the flux inequality for a generalized multi-layer quasi-geostrophic model under  symmetric streamfunction dissipation. In section 4, we consider asymmetric dissipation configurations for the special case of a two-layer quasi-geostrophic model, where we derive various sufficient conditions for satisfying the flux inequality. Conclusions and a brief discussion are given in section 5.

\section{The generalized multilayer model and conservation laws}

Following Ref.~\cite{article:Gkioulekas:p15}, we write the governing equations for the generalized multi-layer model in matrix form:
\begin{align}
&\pderiv{q_\ga}{t}+J(\psi_\ga, q_\ga) = d_\ga  + f_\ga, \\
&d_\ga =  \sum_\gb \ccD_{\ga\gb}\psi_\gb, 
\end{align}
with $J(\gy_\ga, q_\ga)$   the Jacobian of $\gy_\ga$ and $q_\ga$ defined as
\begin{equation}
J(\gy_\ga, q_\ga) = \pderiv{\gy_\ga}{x} \pderiv{q_\ga}{y} - \pderiv{\gy_\ga}{y}\pderiv{q_\ga}{x}.
\end{equation}
Here $\gy_\ga$ represents the streamfunction at the $\ga$-layer, $q_\ga$ represents the potential vorticity at the $\ga$-layer, $\ccD_{\ab}$ is a linear operator encapsulating the dissipation terms, and $f_\ga$ is the forcing term acting on the $\ga$-layer. The index $\ga$ takes the values $\ga=1, 2,\ldots, n$ representing the layer number, for a model involving $n$ layers. Sums over indices, such as in the sum over the index $\gb$ in the dissipation terms above, are assumed to run over all layers $1, 2, \ldots, n$, unless we indicate otherwise. It is also assumed that the streamfunction $\gy_\ga$ and the potential vorticity $q_\ga$ are related via a linear operator $\ccL_{\ab}$ according to:
\begin{equation}
q_\ga (\bfx,t) = \sum_{\gb} \ccL_{\ga\gb} \psi_\gb (\bfx,t).
\end{equation}
The above equations encompass both the two-layer quasi-geostrophic model and the multilayer quasi-geostrophic model, on the assumption that we neglect the $\gb$-effect, arising from the latitudinal dependence of the Coriolis pseudoforce. This is a reasonable assumption for Earth, especially if we restrict our interest to a thin strip of the Earth's surface, oriented parallel to the equator. Waite \cite{article:Waite:2013} has shown  that the relationship between the potential vorticity and streamfunction remains approximately linear in models of stratified turbulence with small buoyancy Reynolds number, but becomes quadratic in the limit of large buoyancy Reynolds number. Baroclinic instability is accounted for by the forcing term $f_\ga$, and implicit in the entire argument is the assumption that it forces the system at large scales only.  This assumption, originally proposed by Salmon \cite{article:Salmon:1978,article:Salmon:1980}, is the only physical assumption implicit in the theoretical framework of the flux inequality, and it has been corroborated numerically \cite{article:Tung:1998,article:Orlando:2003}.

For the sake of simplifying our analysis, we assume that all fields are defined in an infinite two-dimensional domain. Then we can write the streamfunction $\gy$ and the potential vorticity $q$ in terms of their Fourier transforms $\hat\psi_\ga$ and $\hat q _\ga$ as follows:
\begin{align}
\psi_\ga (\bfx,t) &= \int_{\bbR^2} \hat\psi_\ga (\bfk,t) \exp (i\bfk\cdot\bfx) \; \df{\bfk}, \\ 
q_\ga (\bfx,t) &= \int_{\bbR^2} \hat q _\ga (\bfk,t) \exp (i\bfk\cdot\bfx) \; \df{\bfk}.
\end{align}
We assume that the operator $\ccL_{\ab}$ is diagonal in Fourier space. This means that the relation between the streamfunction and the potential vorticity, in Fourier space, reads:
\begin{equation}
\hat q _\ga (\bfk,t) = \sum_{\gb} L_{\ga\gb} (\nrm{\bfk}) \hat\psi_\ga (\bfk,t). \label{eq:QpsiRel}
\end{equation}
Here $\nrm{\bfk}$ represents the 2-norm of the vector $\bfk$. We also assume that $\ccL_{\ab}$ is symmetric with  $L_{\ga\gb}(k) = L_{\gb\ga}(k)$ for all wavenumbers $k$. For quasi-geostrophic models, the matrix $L_{\ga\gb}(k)$ is non-singular for all wavenumbers $k>0$, due to being diagonally dominant, and we assume that to be the case in our abstract formulation given above. Consequently, there is an inverse matrix $L_{\ga\gb}^{-1}(k)$ which defines the inverse operator $\ccL_{\ab}^{-1}$. To accommodate a possible singularity at $k=0$ we assume that at wavenumber $k=0$, in Fourier space, the corresponding field component is $0$ for all fields. This is equivalent to subtracting the mean field and considering only the field fluctuation around the mean.


\subsection{Conservation laws}

We will now show that the generalized layer model, in the absence of dissipation, conserves the total energy $E$ and the total potential enstrophy $G$ under very general conditions on the operator $\ccL_{\ab}$, For any arbitrary scalar field $f(x,y)$ we write the corresponding volume integral using the following notation:
\begin{equation}
\snrm{f} = \iint_{\bbR^2} f(x,y) \;\df{x}\df{y}.\label{eq:VolumeAverage}
\end{equation}
We define the total energy $E$ over all layers, and the layer-by-layer total potential enstrophy $G_\ga$ for layer $\ga$, as $E = -\sum_\ga \snrm{\gy_\ga q_\ga}$ and $G_\ga = \snrm{q_\ga^2}$.  The minus sign ensures that $E$ is positive definite when the operator spectrum $L_{ab}(k)$ satisfies the condition given by Eq.~\eqref{eq:PositiveDefiniteOne} and Eq.~\eqref{eq:PositiveDefiniteTwo} consistently with Eqs.~\eqref{eq:qpsiRelationOne}--\eqref{eq:qpsiRelationThree}  and the sign conventions used by Refs.~\cite{article:Salmon:1978,article:Salmon:1980,article:Gkioulekas:p15}. Specifically, we will show that the potential enstrophy is conserved on a layer-by-layer basis unconditionally regardless of the details of the operator $\ccL_{\ab}$. Conservation of the total energy $E$, over all layers, on the other hand, requires that the operator $\ccL_{\ab}$ be \emph{symmetric} and \emph{self-adjoint}.  To define the self-adjoint property, consider two arbitrary two-dimensional scalar fields $f(x,y)$ and $g(x,y)$. We require that every component of the operator $\ccL_{\ab}$ must satisfy $\snrm{f(\ccL_{\ab} g)} = \snrm{(\ccL_{\ab} f) g}$ for any two fields $f(x,y)$ and $g(x,y)$ for all layer numbers $\ga$ and $\gb$. This self-adjoint property, so defined, follows as an immediate consequence of our previous assumption that the operator $\ccL_{\ab}$ is diagonal in Fourier space. In the proof given below, however, there is no need to use the stronger assumption of diagonality.

The proof is based on the following properties of the nonlinear Jacobian term. If $a(x,y)$ and $b(x,y)$ are two-dimensional smooth scalar-fields that vanish at infinity, then we can show that $\snrm{J(a,b)} =0$, using integration by parts. This result also holds for  the case of fields defined in a finite box with periodic boundary conditions, if the volume integral in Eq.~\eqref{eq:VolumeAverage} is restricted over the box.  Then, we note that, as an immediate consequence of the product rule of differentiation, given three two-dimensional scalar fields $a(x,y), b(x,y)$, and $c(x,y)$ we have
\begin{equation}
\snrm{J(ab,c)} = \snrm{aJ(b,c)} + \snrm{bJ(a,c)} = 0,
\end{equation}
from which we obtain the identity
\begin{equation}
\snrm{aJ(b,c)} = \snrm{bJ(c,a)} = \snrm{cJ(a,b)}.
\label{eq:JacobianCycleIdentity}
\end{equation}
Now, let us go ahead and drop the dissipation and forcing terms and write the time-derivative of the potential vorticity $q_\ga$ as $\dot q_\ga = -J(\gy_\ga, q_\ga)$. Then, the time derivative of the streamfunction $\gy_\ga$ reads:
\begin{equation}
\pderivin{\gy_\ga}{t} = \sum_\gb \ccL_{\ab}^{-1} (\pderivin{q_\gb}{t}) = -\sum_\gb \ccL_{\ab}^{-1} J(\gy_\gb, q_\gb).
\end{equation}

Differentiating the total potential enstrophy $G_\ga$ for the $\ga$ layer with respect to time and employing the identity given by Eq.~\eqref{eq:JacobianCycleIdentity} immediately gives:
\begin{align}
\dderivin{G_{\ga}}{t} &= 2\snrm{q_\ga (\pderivin{q_\ga}{t})} = -2\snrm{q_\ga J(\gy_\ga,q_\ga)} = -2\snrm{\gy_\ga J(q_\ga, q_\ga)} = 0.
\end{align}
Here, we note that from the definition of the Jacobian $J(q_\ga, q_\ga)= 0$. This establishes the layer-by-layer conservation law of potential enstrophy, unconditionally, as claimed. To show the energy conservation law, we differentiate the total energy $E$ with respect to time and obtain:
\begin{align}
\dderivin{E}{t} &= -(\dderivin{}{t})\sum_\ga \snrm{\gy_\ga q_\ga} = -\sum_\ga\snrm{(\pderivin{\gy_\ga}{t}) q_\ga}-\sum_\ga \snrm{\gy_\ga (\pderivin{q_\ga}{t})}\\ 
&= \sum_{\ab} \snrm{q_\ga \ccL_{\ab}^{-1} J(\gy_\gb, q_\gb)}+\sum_\ga\snrm{\gy_\ga J(\gy_\ga, q_\ga)} \\ 
&= \sum_{\ab} \snrm{J(\gy_\gb,q_\gb)\ccL_{\ab}^{-1} q_\ga}+\sum_\ga \snrm{q_\ga J(\gy_\ga, \gy_\ga)}\label{eq:EnergyConservationUseSelfAdjoint} \\ 
&= \sum_{\ab}\snrm{J(\gy_\gb, q_\gb)\ccL_{\ba}^{-1} q_\ga} = \sum_{\gb}\snrm{J(\gy_\gb, q_\gb) \gy_\gb}\label{eq:EnergyConservationUseSymmetric} \\ 
&= \sum_{\gb}\snrm{J(\gy_\gb, \gy_\gb) q_\gb} =0.
\end{align}
 Note that the self-adjoint property is applied at   Eq.~\eqref{eq:EnergyConservationUseSelfAdjoint}, and the symmetric  property is applied at  Eq.~\eqref{eq:EnergyConservationUseSymmetric}. This concludes the  proof.

\subsection{Definition of spectra}

We define spectra for the energy and potential enstrophy using the bracket notation introduced in Ref.~\cite{article:Gkioulekas:p15}. Consider, in general, an arbitrary two-dimensional scalar field $a(x)$. Let $a^{<k}(\bfx)$ be the field obtained from $a(\bfx)$ by setting to zero, in Fourier space, the components corresponding to wavenumbers greater than $k$. Formally, $a^{<k}(\bfx)$ is defined as
\begin{align}
a^{<k}(\bfx) = \int_{\bbR^2} \df{\bfx_0} \int_{\bbR^2} \df{\bfk_0}\;\frac{H(k-\nrm{\bfk_0})}{4\pi^2} \exp (i\bfk_0\cdot (\bfx-\bfx_0))a(\bfx_0),
\end{align}
with $H(x)$ the Heaviside function, defined as the integral of a delta function:
\begin{align}
H(x) &= \casethree{1}{\text{if } x\in(0,+\infty)}{1/2}{\text{if } x=0}{0}{\text{if } x\in (-\infty,0)}.
\end{align}
We now use two filtered fields $a^{<k}(\bfx)$ and $b^{<k}(\bfx)$ to define the bracket $\innerf{a}{b}{k}$ as:
\begin{align}
\innerf{a}{b}{k} &= \dD{k} \int_{\bbR^2} \df{\bfx}\;\avg{a^{<k}(\bfx)b^{<k}(\bfx)} \label{eq:BracketDef} \\ 
&= \frac{1}{2} \int_{A\in\SO{2}} \df{\Omega (A)}\; \avg{[\hat a^{\ast}(kA\bfe) \hat b(kA\bfe) +\hat a(kA\bfe) \hat b^{\ast}(kA\bfe)]}. \label{eq:BracketFourier}
\end{align}
Here, $\hat a (\bfk)$ and $\hat b (\bfk)$ are the Fourier transforms of $a(\bfx)$ and $b(\bfx)$, $\SO{2}$ is the set of all non-reflecting rotation matrices in two dimensions, $d\Omega (A)$ is the measure of a spherical integral, $\bfe$ is a two-dimensional unit vector, and $\avg{\cdot}$ represents  an ensemble average. The star superscript denotes a complex conjugate. Note that Eq.~\eqref{eq:BracketDef} is the definition of the bracket, and  Eq.~\eqref{eq:BracketFourier} follows from Eq.~\eqref{eq:BracketDef} as a consequence. The bracket satisfies the following properties:
\begin{align}
&\innerf{a}{b}{k} = \innerf{b}{a}{k},\\ 
&\innerf{a}{b+c}{k} =\innerf{a}{b}{k} + \innerf{a}{c}{k}, \\ 
&\innerf{a+b}{c}{k} =\innerf{a}{c}{k} + \innerf{b}{c}{k}.
\end{align}
Moreover, every $(\ab)$-component of the operator $\ccL_{\ga\gb}$ is self-adjoint with respect to the bracket:
\begin{equation}
\innerf{\ccL_{\ab} a}{b}{k} = \innerf{a}{\ccL_{\ab} b}{k} = L_{\ab}(k)\innerf{a}{b}{k},
\end{equation}
and the same property is also satisfied by every component of the inverse operator $\ccL_{\ga\gb}^{-1}$:
\begin{equation}
\innerf{\ccL_{\ab}^{-1} a}{b}{k} = \innerf{a}{\ccL_{\ab}^{-1} b}{k} = L_{\ab}^{-1} (k)\innerf{a}{b}{k}.
\end{equation}

Using the bracket, we define the energy spectrum $E(k) = -\sum_\ga \innerf{\gy_\ga}{q_\ga}{k}$, and we also define the layer-by-layer potential enstrophy spectrum $G_{\ga}(k) = \innerf{q_\ga}{q_\ga}{k}$ and the total potential enstrophy spectrum  $G(k) = \sum_\ga G_{\ga}(k)$. Unlike the case of two-dimensional Navier-Stokes, where the enstrophy and energy spectra $G(k)$ and $E(k)$ are related via a simple equation, $G(k)=k^2 E(k)$, in the generalized layer model, the potential enstrophy spectrum and the energy spectrum are related indirectly, as shown below:

 Define the streamfunction spectrum $C_{\ab}(k)= \innerf{\gy_\ga}{\gy_\gb}{k}$. Then, via the properties of the bracket above,  the energy spectrum $E(k)$ reads
\begin{align}
E(k) &= -\sum_{\ga} \innerf{\psi_\ga}{q_\ga}{k} = -\sum_{\ga}\innerf{\psi_\ga}{\sum_{\gb} \ccL_{\ab} \psi_\gb}{k} = -\sum_{\ab}L_{\ab}(k) \innerf{\psi_\ga}{\psi_\gb}{k} \\ 
&= -\sum_{\ab} L_{\ab}(k)C_{\ab}(k), \label{eq:Ek}
\end{align}
and the potential enstrophy spectrum $G_{\ga}(k)$ reads
\begin{align}
G (k) &= \sum_{\ga} \innerf{q_\ga}{q_\ga}{k} = \sum_{\ga} \innerf{\sum_{\gb} \ccL_{\ab} \psi_\gb}{\sum_{\gc} \ccL_{\ac} \psi_\gc}{k} \\ 
&= \sum_{\ab} L_{\ab} (k)\innerf{\psi_\gb}{\sum_{\gc} \ccL_{\ac} \psi_\gc}{k} = \sum_{\abc} L_{\ab} (k) L_{\ac}(k) \innerf{\psi_\gb}{\psi_\gc}{k}\\ 
&= \sum_{\abc} L_{\ab} (k) L_{\ac}(k) C_{\bc}(k). \label{eq:Gk}
\end{align}
Thus, they are related only indirectly via the streamfunction spectrum $C_{\ab}(k)$.

We note that for $\ga\neq\gb$, $C_{\ab}(k)$ may take positive or negative values. For the case $\ga = \gb$ we define $U_{\ga}(k)= \innerf{\gy_\ga}{\gy_\ga}{k}$, which is always positive (i.e., $U_{\ga}(k)\geq 0$), and $U(k) = \sum_{\ga} U_{\ga} (k)$. Then we note that since $U_\ga (k) + U_\gb (k) \pm 2C_{\ab} (k) = \innerf{\gy_\ga \pm \gy_\gb}{\gy_\ga \pm \gy_\gb}{k} \geq 0$, we get the arithmetic-geometric mean inequality $2|C_{\ab}(k)| \leq U_\ga (k) + U_\gb (k)$. We can use this inequality to show that if the matrix $L_{\ab}(k)$ satisfies the diagonal dominance condition
\begin{align}
& L_{\ab}(k)\geq 0, \text{ for } \ga\neq\gb,  \\ 
& \sum_{\gb} L_{\ab}(k) \leq 0,
\end{align}
then the energy spectrum $E(k)$ is  always positive. We give the proof in \ref{app:EkPositiveDefinite}. Both the two-layer quasi-geostrophic model and the multi-layer quasi-geostrophic model  satisfy this diagonal dominance condition. As for the layer-by-layer potential enstrophy spectra $G_{\ga}(k)$, it is immediately obvious that they are unconditionally always positive, regardless of the form of the matrix $L_{\ab}(k)$, since by definition $G_{\ga}(k) = \innerf{q_\ga}{q_\ga}{k}$.

\section{Flux inequality for the $n$-layer model}

We now turn to the main issue of identifying sufficient conditions for satisfying the flux inequality $k^2 \Pi_E (k) - \Pi_G (k) \leq 0$ for quasi-geostrophic models. Let us recall that the energy flux spectrum $\Pi_E (k)$ is defined as the amount of energy transferred from the $(0, k)$ interval to the $(k, +\infty)$ interval per unit time and per unit volume.  Likewise, the potential enstrophy flux spectrum $\Pi_G (k)$ is the amount of potential enstrophy transferred from the $(0,k)$ interval to the $(k, +\infty)$ interval, again per unit time and volume. Assuming a forced-dissipative configuration at steady state and that there is no forcing in the $(k, +\infty)$ wavenumber interval, the energy and potential enstrophy transferred into the $(k, +\infty)$ interval eventually are dissipated somewhere in that interval. It follows that we may write the flux spectra $\Pi_E (k)$ and $\Pi_G (k)$ as integrals of the energy and potential enstrophy dissipation rate spectra $D_E(k)$ and $D_G(k)$:
\begin{align}
\Pi_E (k) &= \int_k^{+\infty} D_E (q) \df{q}, \label{eq:EnergyFluxIntegral} \\ 
\Pi_G (k) &= \int_k^{+\infty} D_G (q) \df{q}, \label{eq:PotentialEnstrophyFluxIntegral}
\end{align}
which implies that
\begin{equation}
k^2 \Pi_E (k) - \Pi_G (k) = \int_k^{+\infty} [k^2 D_E (q) - D_G (q)] \df{q} = \int_k^{+\infty} \gD (k,q) \df{q}. \label{eq:DefinitionOfDelta}
\end{equation}
where $\gD(k,q)$ will be used as an abbreviation for $\gD (k,q)=k^2 D_E (q) - D_G (q)$.  We see that a sufficient condition for establishing the flux inequality is to show that $\gD(k,q) \leq 0$ for all wavenumbers $k < q$. It is also easy to see that $\gD(k,q) > 0$ for all wavenumbers $k_t<k<q$ is sufficient for establishing the violation of the flux inequality for all wavenumbers $k>k_t$.

For the case of two-dimensional Navier-Stokes turbulence, the dissipation rate spectra $D_E(k)$ and $D_G(k)$ are related via $D_G(k) = k^2 D_E(k)$. This immediately gives $\gD (k,q) = k^2 D_E (q) - D_G (q) = (k^2-q^2) D_E (q)\leq 0$ for all wavenumbers $k<q$ (since $D_E(k)\geq 0$), which in turn gives the flux inequality $k^2 \Pi_E (k) - \Pi_G (k) \leq 0$. The physical interpretation of this inequality is that when we stretch the separation of scales in the downscale range, the energy dissipation rate at small-scales vanishes rapidly. As a result, most of the injected energy cannot cascade downscale although, as noted previously \cite{article:Tung:2005,article:Tung:2005:1}, a small amount of energy is able to do so. As we have seen in the previous section, for the case of quasi-geostrophic models, the energy and potential enstrophy dissipation rate spectra no longer have a direct and simple relation with each other, so the validity of the flux inequality needs to be carefully re-examined.

For the general multi-layer quasi-geostrophic model, the relationship between the potential vorticities $q_\ga$ and the streamfunctions $\gy_\ga$ is given by 
\begin{align}
q_1 &= \del^2 \gy_1 + \mu_1 k_R^2 (\gy_2-\gy_1), \label{eq:qpsiRelationOne} \\
q_\ga &= \lapl\gy_\ga - \gl_\ga k_R^2 (\gy_\ga-\gy_{\ga-1})+\mu_\ga k_R^2 (\gy_{\ga+1}-\gy_\ga), \text{ for } 1<\ga<n, \label{eq:qpsiRelationTwo}\\
q_n &= \lapl\gy_n-\gl_n k_R^2 (\gy_n-\gy_{n-1}). \label{eq:qpsiRelationThree}
\end{align}
Here, $k_R$ is the Rossby wavenumber and $\gl_\ga$ and $\mu_\ga$ are the non-dimensional Froude numbers, given by
\begin{align*}
\gl_\ga &= \frac{1}{2}\frac{h_1}{h_\ga}\frac{\rho_2-\rho_1}{\rho_\ga-\rho_{\ga-1}}, \text{ for } 1<\ga \leq n, \\
\mu_\ga &= \frac{1}{2}\frac{h_1}{h_\ga}\frac{\rho_2-\rho_1}{\rho_{\ga+1}-\rho_\ga}, \text{ for } 1\leq\ga < n,
\end{align*}
with $\gr_\ga$ the average density of layer $\ga$, and $h_\ga$ the average height of layer $\ga$ (in pressure coordinates). The definition of the non-dimensional Froude numbers was adjusted with a $1/2$ numerical factor, from the one given by Evensen \cite{article:Evensen:1994}, to ensure agreement with the formulation of the two-layer quasi-geostrophic model given by Salmon \cite{article:Salmon:1980} for the case $n=2$. The components of the corresponding matrix $L_{\ab}(k)$ are given by 
\begin{align*}
L_{\ga\ga}(k) &= \casethree{-k^2-\mu_1 k_R^2}{\ga=1}{-k^2-(\gl_\ga+\mu_\ga) k_R^2}{1<\ga <n}{-k^2-\gl_n k_R^2}{\ga=n,} \\
L_{\ga, \ga+1}(k) &= \mu_\ga k_R^2, \text{ for } 1\leq\ga <n, \\
L_{\ga, \ga-1}(k) &= \gl_\ga k_R^2, \text{ for } 1< \ga \leq n.
\end{align*}
In the present paper we limit ourselves to the special case of a symmetrically coupled multi-layer quasi-geostrophic model, where we assume that the layer thickness $h_\ga$ is the same for all layers, thereby yielding a symmetric matrix $L_{\ab}(k)$ such that $L_{\ga, \ga+1}(k) = L_{\ga+1, \ga}(k)$ for all $1\leq\ga < n$. 

To consider the flux inequality for this general $n$-layer model, we begin with writing the dissipation rates $D_E(k)$ and $D_G(k)$ for the energy and potential enstrophy  in terms of the streamfunction spectrum $C_{\ab}(k)$. We assume that the dissipation operation $\ccD_{\ga\gb}$ is  diagonal in Fourier space and that the Fourier transform of the dissipation term $\ccD_{\ga\gb} \psi_\gb$ reads:
\begin{equation}
(\ccD_{\ab} \gy_\gb) (\bfx, t) = \int_{\bbR^2} D_{\ab} (\nrm{\bfk}) \hat \gy_\gb (\bfk,t) \exp (i\bfk\cdot\bfx) \;\df{\bfk}.
\end{equation}
Then, in  \ref{app:DissipationRateSpectra} we show that  the energy dissipation rate spectrum $D_E (k)$ and the layer-by-layer potential enstrophy dissipation rate spectra $D_{G_\ga} (k)$ are given by
\begin{align}
D_E (k) &= 2 \sum_{\ab} D_{\ab} (k) C_{\ab} (k), \label{eq:DissEk}   \\ 
D_{G_\ga} (k) &= -2 \sum_{\bc} L_{\ab}(k) D_{\ac} (k) C_{\bc} (k). \label{eq:DissGk}
\end{align}
Note that in order for the dissipation terms to be truly dissipative, the dissipation spectra $D_E (k)$ and $D_G(k)$ need to be both always positive for all wavenumbers $k$.  From the general form of the above equations this is not readily obvious. However, for simpler configurations of the dissipation operators, the above expressions for $D_E (k)$ and $D_G(k)$ simplify considerably, thereby making it possible to establish that they are both always positive. These expressions also underscore the main difference between two-dimensional Navier-Stokes turbulence and quasi-geostrophic turbulence and the reason why the flux inequality becomes a non-trivial problem in the latter case. Unlike two-dimensional turbulence, and in spite of the twin conservation laws of energy and potential enstrophy, the dissipation rates $D_E(k)$ and $D_G(k)$ are no longer related by any simple relation of the form  $D_G(k) = k^2D_E(k)$.

We restrict our attention to the case where  the dissipation operators at every layer involve only the streamfunction of the corresponding layer, with no explicit interlayer terms. This can be arranged in terms of a linear operator $\ccD_{\ga}$ applied to  the streamfunction $\gy_\ga$. If $D_\ga (k)$ is the spectrum of the positive-definite operator $\ccD_{\ga}$, then for the case of a dissipation term $d_\ga = \ccD_{\ga} \gy_\ga$, we have $D_{\ab}(k)=\gd_{\ab} D_\gb (k)$, with $\gd_{\ab}$ given by
\begin{equation}
\gd_{\ab} = \casetwo{1}{\ga = \gb}{0}{\ga\neq\gb}.
\end{equation}  
We designate this case as \emph{streamfunction-dissipation}. The $D_E (k)$ and $D_{G_\ga}(k)$ simplify as:
\begin{align}
D_E (k) &= 2 \sum_{\ab} D_{\ab} (k) C_{\ab} (k) = 2 \sum_{\ab} \gd_{\ab} D_\gb (k) C_{\ab} (k) = 2 \sum_{\ga} D_\ga (k) C_{\ga\ga} (k) = 2 \sum_{\ga} D_\ga (k) U_{\ga} (k), \label{eq:EkStrDiss} \\ 
D_{G_\ga}(k) &= -2 \sum_{\bc} L_{\ab}(k) D_{\ac} (k) C_{\bc} (k) = -2 \sum_{\bc} L_{\ab}(k) \gd_{\ac} D_\gc (k) C_{\bc} (k) = -2 \sum_{\gb} L_{\ab}(k) D_\gb (k) C_{\ab} (k). \label{eq:GkStrDiss}
\end{align}
Note that for $D_{\ga}(k)\geq 0$, it follows that $D_E (k)\geq 0$, but it is not obvious that the same result extends to $D_{G_\ga}(k)$. However,  if we further assume that the same operator is used for all layers, i.e. $D_\ga (k) = D(k)$, then we have the more specialized case of \emph{symmetric streamfunction-dissipation}, and the dissipation rate spectra $D_E (k)$ and $D_{G}(k)$ can be simplified further to give:
\begin{align}
D_E (k) &= 2 \sum_{\ga} D_\ga (k) U_{\ga} (k) = 2 D(k) \sum_{\ga} U_{\ga} (k) = 2 D(k) U(k), \\ 
D_{G}(k) &= \sum_{\ga} D_{G_\ga}(k) = -2 \sum_{\ab} L_{\ab}(k) D_\gb (k) C_{\ab} (k) = 2 D(k) \left[-\sum_{\ab} L_{\ab}(k) C_{\ab} (k) \right] = 2 D(k) E(k).
\end{align}
Now, $D(k)\geq 0$ implies both $D_E (k) \geq 0$ and $D_G (k) \geq 0$. 

It follows that, under symmetric streamfunction dissipation, $\gD(k,q)$ is  given by
\begin{equation}
\gD (k,q) = k^2 D_E (q) - D_G (q) = k^2 D(q) U(q) - D(q) E(q) = D(q) [k^2 U(q) - E(q)],
\end{equation}
and since $D(q)\geq 0$, the validity of the flux inequality is dependent on the sign of the factor $k^2 U(q) - E(q)$. That sign is in turn intimately related with the expression $\gc_\ga (k,q)$ defined as:
\begin{equation}
\gc_\ga (k,q) = k^2+\sum_\gb L_{\ab} (q).
\end{equation}
Note that for the case of two-dimensional Navier-Stokes, $L(q)$ becomes a $1\times 1$ matrix with $L_{11}(q) = q^2$, thus $\gc_\ga (k,q) = k^2-q^2$, which is negative when $k<q$. For more generalized $n$-layer quasi-geostrophic models, the expression $\gc_\ga(k,q)$  continues to be given by $\gc_\ga (k,q) = k^2-q^2$ which remains negative when $k<q$ for all layers $\ga$. We will now show that:
\begin{proposition}
In a generalized $n$-layer model, under symmetric streamfunction dissipation $d_\ga = +\ccD \gy_\ga$ with spectrum $D(k)\geq 0$, we assume that $L_{\ab}(q)\geq 0$ when $\ga\neq\gb$, and $L_{\ab}(q)=L_{\ba}(q)$, and $\gc_\ga(k,q)\leq 0$ when $k<q$ for all $\ga$. It follows that:
\begin{equation*}
\gD (k,q) \leq D(q) \sum_\ga \gc_\ga (k,q) U_\ga (q) \leq 0.
\end{equation*}
\end{proposition}

\begin{proof}
We begin by recalling from  \ref{app:EkPositiveDefinite}, that $E(q)$ can be rewritten as 
\begin{equation}
E(q) =  -\sum_{\ab} L_{\ab}(q)  U_\ga (q)  - \frac{1}{2}\sum_{\substack{\ab \\ \ga\neq\gb}} L_{\ab}(q) [2C_{\ab}(q)-U_{\ga}(q)-U_{\gb}(q)].
\end{equation}
It follows that $k^2 U(q) - E(q)$ satisfies:
\begin{align}
k^2 U(q) - E(q) &= k^2 \sum_{\ga} U_\ga (q) + \sum_{\ab} L_{\ab}(q) U_\ga (q) + \frac{1}{2}\sum_{\substack{\ab \\ \ga\neq\gb}} L_{\ab}(q) [2C_{\ab}(q)-U_{\ga}(q)-U_{\gb}(q)] \\
&\leq k^2 \sum_{\ga} U_\ga (q) + \sum_{\ab} L_{\ab}(q) U_\ga (q) = \sum_{\ga} \biggr( k^2 + \sum_{\gb} L_{\ab}(q) \biggl) U_\ga (q) \\
&= \sum_\ga \gc_\ga (k,q) U_\ga (q).
\end{align}
The inequality uses the assumption $L_{\ab}(q)\geq 0$ combined with the arithmetic-geometric mean  inequality $2C_{\ab}(q) \leq U_\ga (q) + U_\gb (q)$ of the streamfunction spectra. It follows that
\begin{equation}
\gD (k,q) = D(q) [k^2 U(q) - E(q)] \leq D(q) \sum_\ga \gc_\ga (k,q) U_\ga (q) \leq 0,
\end{equation}
since $D(q)\geq 0$, $U_\ga (q)\geq 0$, and $\gc_\ga (k,q)\leq 0$, thereby concluding the proof. 
\end{proof}

The above result establishes the unconditional validity of the flux inequality for generalized $n$-layer quasi-geostrophic models under   symmetric streamfunction dissipation. We note that the condition $L_{\ab}(q)\geq 0$ is needed to establish that the energy spectrum $E(k)$ is always positive, and all physically relevant quasi-geostrophic models will also satisfy the condition $\gc_\ga(k,q)\leq 0$ for all $k<q$.  As we have already argued, for any general $n$-layer quasi-geostrophic model, we have $\gc_\ga(k,q)=k^2-q^2$ for all layers $\ga$, so the assumption is mathematical and does not impose any physical constraints in the model's formulation. No other restrictions are needed by the above proposition. Physically, this means that under symmetric streamfunction dissipation, the behavior of any generalized $n$-layer model will be similar to two-dimensional turbulence, where the subdominant downscale energy cascade is too weak to cause a transition from $k^{-3}$ scaling to $k^{-5/3}$ scaling in the downscale inertial range.

\section{Flux inequality in a two-layer model}

The previous results, derived over a general $n$-layer quasi-geostrophic model also apply to the special case of a two-layer quasi-geostrophic model. Consequently, the flux inequality will be satisfied by any two-layer quasi-geostrophic models under symmetric streamfunction dissipation. We will now concentrate on investigating the validity of the flux inequality in two-layer quasi-geostrophic models with asymmetric dissipation.

Since the details of the argument  below are very technical, we provide a brief outline. In section 4.1 we write the governing equations for the two-layer quasi-geostrophic model and define the two novel features of the proposed configuration of the dissipation terms: \emph{extrapolated Ekman damping}, controlled by the parameter $\mu$, and \emph{small-scale differential dissipation}, which is controlled by the parameter $\gD\nu$. In section 4.2 we derive the general form of the energy dissipation spectrum $D_E (k)$ and the potential enstrophy dissipation spectrum $D_G (k)$ for the most general dissipation term configuration. In section 4.3 we derive Proposition~\ref{prop:GeneralSuffCondition}, which gives a sufficient condition, via Eq.~\eqref{eq:GeneralSuffCondition}, for satisfying the flux inequality, in terms of the dissipation term configuration, which is completely described by the spectra $D_1 (k)$, $D_2 (k)$, $d(k)$, and the parameter $\mu$. The proposition is very abstract and general, as it accounts for a very wide range of possible  configurations. In section 4.4 we derive, from proposition~\ref{prop:GeneralSuffCondition}, a series of corollaries for four special cases of interest: (a) the case of streamfunction dissipation with both extrapolated Ekman damping and differential small-scale dissipation,  given by Eq.~\eqref{eq:SuffCondOne}; (b) the case of streamfunction dissipation with differential  small-scale dissipation but without extrapolated Ekman dissipation, given by Eq.~\eqref{eq:SuffCondTwo}; (c) the case of streamfunction dissipation with extrapolated Ekman dissipation but without differential small-scale dissipation,  given by Eq.~\eqref{eq:SuffCondThree}; (d) the case of the standard symmetric streamfunction dissipation without any special features, given by Eq.~\eqref{eq:OrigSuffCond}. A careful comparison is given between the sufficient conditions to satisfy the flux inequality for each of the four cases. Finally, section 4.5 gives a different set of sufficient conditions to satisfy the flux inequalities in terms of the streamfunction spectra. Future work should combine these conditions with some phenomenological model of the energetics of the two-layer model to extract useful information. 

\subsection{Model formulation}

The two-layer quasi-geostrophic model can be formulated in terms of two potential vorticity equations of the form 
\begin{align}
\pderiv{q_1}{t} &+ J(\gy_1, q_1) = f_1 + d_1, \label{eq:PotVortOne}  \\
\pderiv{q_2}{t} &+ J(\gy_2, q_2) = f_2 + d_2, \label{eq:PotVortTwo}
\end{align}
with the relationship between the potential vorticities $q_1$, $q_2$ and the streamfunctions $\gy_1$, $\gy_2$ given by 
\begin{align}
q_1 &=  \del^2 \gy_1  + \frac{k_R^2}{2}(\gy_2-\gy_1), \label{eq:DefPotVortOne} \\
q_2 &=  \del^2 \gy_2  - \frac{k_R^2}{2}(\gy_2-\gy_1). \label{eq:DefPotVortTwo} 
\end{align}
Here $q_1$, $\gy_1$ correspond to the top layer and $q_2$, $\gy_2$ correspond to the bottom layer. As explained in the introduction, we situate the top layer at  $p_1=0.25$Atm and the bottom layer at $p_2=0.75$Atm. In terms of the generalized layer model, Eq.~\eqref{eq:DefPotVortOne} and Eq.~\eqref{eq:DefPotVortTwo} correspond to an operator $\cL_{\ab}$ with spectrum $L_{\ab}(k)$ given by
\begin{equation}
L(k) = -\mattwo{a(k)}{b(k)}{b(k)}{a(k)},
\end{equation}
with $a(k)$ and $b(k)$ given by $a(k) = k^2+k_R^2/2$ and $b(k) = -k_R^2$. Using differential hyperdiffusion at the small scales and extrapolated Ekman dissipation at the bottom layer gives
\begin{align}
d_1 &= \nu (-1)^{p+1} \del^{2p+2} \gy_1, \label{eq:DissOne} \\
d_2 &= (\nu + \gD\nu) (-1)^{p+1} \del^{2p+2}\gy_2 - \nu_E \del^2 \gy_s. \label{eq:DissTwo}
\end{align}
Here we assume that the hyperdiffusion is stronger at the lower layer, with $\gD\nu >0$ being the additional hyperdiffusion coefficient added to the lower-layer (the reader should not confuse the coefficient $\gD\nu$ with the previously defined function $\gD (k,q)$). Furthermore, the Ekman term is given in terms of the streamfunction $\gy_s$ at the Ekman surface layer  which is linearly extrapolated from $\gy_1$ and $\gy_2$ and it is given by $\gy_s = \gl\gy_2+\mu\gl\gy_1$, with $\gl$ and $\mu$ given by
\begin{equation}
\gl = \frac{p_s-p_1}{p_2-p_1} \text{ and } \mu = \frac{p_2-p_s}{p_s-p_1}.
\end{equation}
In other words, $\gy_s$ is defined so that, plotted on a pressure-streamfunction plane, the three points $(p_s,\gy_s), (p_1, \gy_1), (p_2, \gy_2)$ are  colinear. Using $p_1=0.25$Atm, $p_2=0.75$Atm and $p_s = 1$Atm gives $\gl=3/2$ and $\mu=-1/3$. It is worth noting that for any arbitrary placement of the top and bottom layer that satisfies $0<p_1<p_2<p_s$, we can show that $-1<\mu<0$. This constraint on $\mu$  is all that is needed to derive the main results of this paper, so the precise placement of the surface layer is not important for our argument below. On the other hand, moving the potential vorticity layers around would necessitate non-symmetric generalizations of the operator $\ccL_{\ab}$, which may be interesting for oceanographic modeling, but not necessary for atmospheric modeling, and beyond the scope of this paper. We will therefore assume that $p_1$ and $p_2$ are fixed but allow $p_s$ to vary as $p_2 \leq p_s \leq1$Atm, which in turn corresponds to $-1/3 \leq \mu \leq 0$.

The dissipation term configuration given by Eq.~\eqref{eq:DissOne} and Eq.~\eqref{eq:DissTwo} corresponds to setting the generalized dissipation operator spectrum $D_{\ab}(k)$ equal to 
\begin{equation}
D(k) = \mattwo{D_1(k)}{0}{\mu d(k)}{D_2 (k)+d(k)}, \label{eq:GeneralDissipationMatrix}
\end{equation}
with $D_1(q)$, $D_2(q)$, and $d(q)$ given by
\begin{equation}
D_1(k) = \nu k^{2p+2}  \text{ and } D_2 (k) = (\nu +\gD\nu) k^{2p+2}  \text{ and } d(k) = \gl \nu_E k^2.
\end{equation}
Note that for $\mu=0$ and $\gl=1$, this reduces to the simpler case of streamfunction dissipation.

\subsection{Dissipation rate spectra for the two-layer model}

 We may now leverage Eq.~\eqref{eq:DissEk} and Eq.~\eqref{eq:DissGk} to calculate the energy and potential enstrophy dissipation rate spectrum $D_E(k)$ and $D_G(k)$ in terms of the streamfunction spectra  $U_1(k)$, $U_2(k)$, and $C_{12}(k)$. For the case of the energy dissipation rate spectrum $D_E (k)$, noting that $D_{12}(k)=0$, a simple calculation gives 
\begin{align}
D_E (k) &= 2D_{11}(k) U_1(k) + 2D_{22}(k) U_2 (k) + 2D_{21}(k) C_{21}(k) \\
&= 2D_{11}(k)  U_1(k) + 2D_{22}(k) U_2(k) + D_{21}(k) [2C_{12}(k)-U(k)] + D_{21}(k) U(k)  \\
&= [2D_{11}(k) + D_{21}(k)] U_1(k)  + [2D_{22}(k) +D_{21}(k)] U_2(k) + D_{21}(k) [2C_{12}(k)-U(k)] \\
&= A_E^{(1)}(k) U_1(k)  + A_E^{(2)}(k) U_2(k) + A_E^{(3)}(k)[2C_{12}(k)-U(k)],
\end{align}
with $A_E^{(1)}(k)$, $A_E^{(2)}(k)$, and $A_E^{(3)}(k)$ given by
\begin{align}
A_E^{(1)}(k) &= 2D_{11}(k) + D_{21}(k) = 2D_1(k)  + \mu d(k), \\
A_E^{(2)}(k) &= 2D_{22}(k) +D_{21}(k) = 2D_2(k)  + 2d(k)  +\mu d(k), \\
A_E^{(3)}(k) &= D_{21}(k) = \mu d(k).
\end{align}
 We  note that terms involving the streamfunction cross-spectrum $C_{12}(k)$ have been reorganized in terms of $2C_{12}(k)-U(k)$  so that we can take advantage of the inequality $2C_{12}(k)-U(k) \leq 0$. For the potential enstrophy dissipation rate spectrum $D_G(k)$, we take advantage of the symmetry assumption $L_{\ab}(k)=L_{\ba}(k)$ to rewrite Eq.~\eqref{eq:DissGk} as 
\begin{align}
D_G (k) &= -2\sum_{\abc} L_{\ab}(k)  D_{\ac}(k)  C_{\bc}(k) = -2\sum_{\abc} L_{\ba}(k)  D_{\ac}(k)  C_{\bc}(k) = -2\sum_{\bc} (LD)_{\bc}(k) C_{\bc}(k).
\end{align}
The components of $(LD)(k)$ are given by
\begin{align}
(LD)(k) &= -\mattwo{a(k)}{b(k)}{b(k)}{a(k)} \mattwo{D_1(k)}{0}{\mu d(k)}{D_2 (k)+d(k)} \\
&= -\mattwo{a(k) D_1(k) + \mu b(k) d(k)}{b(k)  [D_2(k)  + d(k)]}{b(k)  D_1(k)  + \mu a(k)  d(k) }{a(k) [D_2(k)  + d(k)]},
\end{align}
and it follows that $D_G(k)$ is given by
\begin{align}
D_G (k) &= -2\{(LD)_{11}(k)  U_1(k)  + (LD)_{22}(k)  U_2(k)  + [(LD)_{12}(k) +(LD)_{21}(k) ] C_{12}(k)  \} \\
&= -\{[2(LD)_{11}(k)+(LD)_{12}(k)+(LD)_{21}(k)] U_1(k) + [2(LD)_{22}(k)+(LD)_{12}(k)+(LD)_{21}(k)] U_2(k)  \nonumber \\ &\quad + [(LD)_{12}(k)+(LD)_{21}(k)] [2C_{12}(k)  - U(k)] \} \\
&= A_G^{(1)}(k) U_1(k)  + A_G^{(2)}(k) U_2(k) + A_G^{(3)}(k)[2C_{12}(k)-U(k)],
\end{align}
with $A_G^{(1)}(k)$, $A_G^{(2)}(k)$, $A_G^{(3)}(k)$ given by 
\begin{align}
A_G^{(1)}(k)  &= -[2(LD)_{11}(k)+(LD)_{12}(k)+(LD)_{21}(k)] \\
&= 2[a(k)  D_1(k)  + \mu b(k)  d(k) ] + b(k)  D_1(k)  + \mu a(k)  d(k)  +b(k)  [D_2(k) +d(k) ] \\
&= [2a(k) +b(k) ]D_1(k)  + b(k)  D_2(k)  + [2b(k) +a(k) ]\mu d(k)  + b(k)  d(k),  \\
A_G^{(2)}(k) &= -[2(LD)_{22}(k)+(LD)_{12}(k)+(LD)_{21}(k)] \\
&= 2a(k)  [D_2(k)  + d(k)  ]+b(k)  D_1(k)  + \mu a(k)  d(k)  + b(k)  [D_2(k)  + d(k) ]  \\
&= b(k)  D_1(k)  + [2a(k) +b(k) ] D_2(k)  + [2a(k) +b(k) ]d(k)  + \mu a(k) d(k), \\
A_G^{(3)}(k) &= -[(LD)_{12}(k)+(LD)_{21}(k)] = b(k)  D_1(k)  + \mu a(k)  d(k)  + b(k) [D_2(k)  + d(k)] \\
&= b(k)  [D_1(k) + D_2(k) ]+b(k)  d(k)  + \mu a(k)  d(k). 
\end{align}
The above expressions for $D_E(k)$ and $D_G(k)$ are the point of departure for the investigation of the flux inequality under the general case of streamfunction dissipation with extrapolated Ekman dissipation and differential small-scale dissipation. 

\subsection{Sufficient conditions in terms of dissipation coefficients}

As we have discussed previously, to satisfy the flux inequality $k^2 \Pi_E (k) - \Pi_G (k) \leq 0$ for a given wavenumber $k$, it is sufficient to show that $\gD (k,q)\leq 0$ for all wavenumbers $k<q$. Using our previous expressions for the energy dissipation rate $D_E(k)$ and the potential enstrophy dissipation rate $D_G(k)$, we can calculate $\gD (k,q)$. Consequently,  $\gD(k,q)$ is given by
\begin{align}
\gD (k,q) &= k^2 D_E (q) - D_G(q) \\
&= A_1 (k,q) U_1(q) +  A_2 (k,q) U_2(q) +  A_3 (k,q) [2C_{12}(q) - U(q)],
\end{align}
with $A_1(k,q)$, $A_2(k,q)$, and $A_3(k,q)$ given by 
\begin{align}
A_1 (k,q) &= k^2 A_E^{(1)}(q) - A_G^{(1)}(q) \\
&=k^2 [2D_1(q) + \mu d(q)]-[2a(q) +b(q)]D_1(q) - b(q) D_2(q)  - [2b(q)+a(q)]\mu d(q) - b(q) d(q) \\
&=[2k^2-2a(q)-b(q)]D_1(q) - b(q) D_2(q) - b(q) d(q) + [k^2-2b(q)-a(q)]\mu d(q), \\
A_2 (k,q) &= k^2 A_E^{(2)}(q) - A_G^{(2)}(q) \\
&= k^2[2D_2(q)+\mu d(q) +2d(q)] - b(q) D_1(q) - [2a(q)+b(q)]D_2(q)  \nonumber \\ &\quad -[2a(q)+b(q)]d(q) - \mu a(q) d(q) \\
&= -b(q) D_1(q)+[2k^2-2a(q)-b(q)]D_2(q) + [2k^2-2a(q)-b(q)]d(q)  + \mu [k^2-a(q)]d(q), \\
A_3 (k,q) &= k^2 A_E^{(3)}(q) - A_G^{(3)}(q) = k^2 \mu d(q) - b(q)[D_1(q)+D_2(q)]-b(q) d(q) - \mu a(q) d(q) \\
&= -b(q) [D_1(q)+D_2(q)]-b(q) d(q) + \mu [k^2-a(q)]d(q).
\end{align}
We observe that $U_1(q)\geq 0$ and $U_2(q)\geq 0$ and $2C_{12}(q)-U(q)\leq 0$, consequently the sign of $\gD(k,q)$ depends on the sign of the coefficients $A_1(k,q)$, $A_2(k,q)$, and $A_3(k,q)$.  For the argument below we may assume that $-1<\mu<0$ and $D_1(q) \leq D_2(q)$. Here $D_1(q) < D_2(q)$ corresponds to differential small-scale diffusion (i.e $\gD\nu >0$) and $D_1(q)=D_2(q)$ corresponds to symmetric small-scale dissipation (i.e $\gD\nu = 0$). We begin our argument with the following lemma:

\begin{lemma}
Assume that $b(q)<0$ and $k^2-a(q)-b(q)<0$. Assume also streamfunction dissipation with both differential small-scale dissipation and extrapolated Ekman dissipation with $-1<\mu<0$. Then $A_3 (k,q)\geq 0$, and furthermore, if $D_1(q)\leq D_2(q)$, then we also have $A_2 (k,q)\leq 0$. 
\label{prop:theLemma}
\end{lemma}

\begin{proof}
We recall that $A_3 (k,q)$ is given by
\begin{equation}
A_3 (k,q) = -b(q) [D_1(q)+D_2(q)]-b(q) d(q) + \mu [k^2-a(q)]d(q).
\end{equation}
Since, by definition, $D_1(q)\geq  0$, and $D_2(q)\geq 0$, and $d(q)\geq 0$, and since $b(q)<0$, and $k^2-a(q)=[k^2-a(q)-b(q)]+b(q)<k^2-a(q)-b(q)<0$, and $\mu <0$, it follows that all contributing terms to $A_3(k,q)$ are positive and therefore $A_3 (k,q)\geq 0$. For the case of $A_2(k,q)$, let us assume first that $D_1(q)\leq D_2(q)$. We rewrite $A_2(k,q)$ as follows:
\begin{align}
A_2 (k,q) &= -b(q) D_1(q)+[2k^2-2a(q)-b(q)]D_2(q) + [2k^2-2a(q)-b(q)]d(q)  + \mu [k^2-a(q)]d(q)\\
&= -b(q) [D_1(q)-D_2(q)]+2[k^2-a(q)-b(q)]D_2(q) \nonumber \\ &\quad   + 2[k^2-a(q)-b(q)]d(q) + b(q)d(q)+\mu [k^2-a(q)-b(q)]d(q)+\mu b(q)d(q) \\
&= -b(q)[D_1(q)-D_2(q)]+[k^2-a(q)-b(q)][2D_2(q)+(2+\mu) d(q)]   + (\mu+1)b(q)d(q).
\end{align}
Since $b(q)<0$, and $k^2-a(q)-b(q)<0$,  and $\mu+2>0$, and $\mu+1>0$, we see that all contributing terms to $A_2(k,q)$ are negative and therefore $A_2(k,q)\leq 0$. This concludes the proof.
\end{proof}

\begin{proposition}
Assume streamfunction dissipation with both differential small-scale dissipation and extrapolated Ekman dissipation with $-1<\mu<0$. Assume also that $k^2-a(q)-b(q)<0$, and $b(q)<0$, and $D_1 (q) \geq 0$, and $D_2 (q) \geq 0$, and $\gD D(q)\equiv D_2(q)-D_1(q)\geq 0$, and also that $D_1(q)$, $\gD D(q)$, and $d(q)$ satisfy 
\begin{equation}
\frac{2D_1(q) + \mu d(q)}{[D_2(q)-D_1(q)]+(\mu+1)d(q)} > \frac{b(q)}{k^2-a(q)-b(q)}. \label{eq:GeneralSuffCondition}
\end{equation}
Then it follows that $\gD(k,q)\leq 0$. 
\label{prop:GeneralSuffCondition}
\end{proposition}

\begin{proof}
We recall that $\gD(k,q)$ is given by 
\begin{align}
\gD (k,q) &= A_1 (k,q) U_1(q) +  A_2 (k,q) U_2(q) +  A_3 (k,q) [2C_{12}(q) - U(q)].
\end{align}
Using the previous lemma, from the given assumptions above, we have $A_2(k,q)\leq 0$ and $A_3(k,q)\geq 0$. Now let us rewrite $A_1(k,q)$ as
\begin{align}
A_1(k,q) &=[2k^2-2a(q)-b(q)]D_1(q) - b(q) D_2(q) - b(q) d(q)   + [k^2-2b(q)-a(q)]\mu d(q)\\
&= 2[k^2-a(q)-b(q)]D_1(q)-b(q)[D_2(q)-D_1(q)]-b(q)d(q) \nonumber \\ &\quad +[k^2-a(q)-b(q)]\mu d(q)-\mu b(q)d(q) \\
&= [k^2-a(q)-b(q)][2D_1(q)+\mu d(q)]-b(q)[D_2(q)-D_1(q)+(\mu+1)d(q)].
\end{align}
Since $k^2-a(q)-b(q)<0$ and $D_2(q)-D_1(q)+(\mu+1)d(q)>0$, it follows that $A_1(k,q)\leq 0$ if and only if Eq.~\eqref{eq:GeneralSuffCondition} is satisfied. Thus, since we also know that $U_1(q)\geq 0$, and $U_2(q)\geq 0$, and $2C_{12}(q)-U(q)\leq 0$, it follows that all terms contributing to $\gD (k,q)$  are negative, and therefore $\gD (k,q)\leq 0$.
\end{proof}

\subsection{Discussion of sufficient conditions in terms of dissipation coefficients}

We will now use Proposition~\ref{prop:GeneralSuffCondition} to extract sufficient conditions to satisfy the flux inequality for the four dissipation term  configurations, outlined in the beginning of this section, in terms of the dissipation term coefficients. Our goal is to explore the restrictiveness of these conditions for each configuration.  For the first dissipation configuration, we consider streamfunction dissipation with both differential small-scale dissipation and extrapolated Ekman dissipation, with the surface layer placed at $1$Atm. Mathematically, this corresponds to using $D_1(q)=\nu q^{2p+2}$, and $d(q)=(3/2)\nu_E q^2$ (since $\gl=3/2$), and $D_2(q)-D_1(q)=\gD\nu q^{2p+2}$, and $\mu=-1/3$.  It is easy to show that, given these choices, Proposition~\ref{prop:GeneralSuffCondition} gives the following statement
\begin{equation}
0 < \frac{\gD\nu q^{2p}+\nu_E}{4\nu q^{2p}-\nu_E} < \frac{q^2-k^2}{k_R^2} \implies \gD (k,q) \leq  0. \label{eq:SuffCondOne}
\end{equation}
Note that the hypothesis  given by Eq.~\eqref{eq:SuffCondOne} requires that $\nu_E <4\nu q^{2p}$, which ensures that both sides of Eq.~\eqref{eq:GeneralSuffCondition} are positive. We may then invert both sides of Eq.~\eqref{eq:GeneralSuffCondition} in the process of  obtaining Eq.~\eqref{eq:SuffCondOne}. On the other hand, for $\nu_E >4\nu q^{2p}$,  Eq.~\eqref{eq:GeneralSuffCondition}  is violated, as its left-hand side becomes negative while its right-hand side remains positive. More precisely, in Eq.~\eqref{eq:GeneralSuffCondition}, the right-hand side is positive for $q>k$, the denominator of the left-hand side satisfies $D_2(q)-D_1(q)+(\mu+1)d(q)>0$ by the given choices for $D_1 (q)$,  $D_2 (q)$, $d(q)$, and $\mu$, and the constraint $\nu_E <4\nu q^{2p}$ is needed to ensure that the numerator $2D_1(q) + \mu d(q)$ is not negative, so that it can be possible for Eq.~\eqref{eq:GeneralSuffCondition} to be satisfied. Consequently, we see that increasing either $\nu_E$ or $\gD\nu$ indicates a tendency towards violating the flux inequality. The role of differential diffusion is very important here since, for $\gD\nu>0$, the left-hand-side of the hypothesis in Eq.~\eqref{eq:SuffCondOne} will approach $\gD\nu/(4\nu)$ and remain bounded for large wavenumbers $q$, whereas for $\gD\nu=0$, the same left-hand-side will vanish rapidly to zero with increasing wavenumber $q$. As a result, violating the flux inequality may become easier under differential small-scale dissipation $\gD\nu$. On the other hand, the role of $\nu_E$ becomes even more dramatic, since increasing $\nu_E$ from $0$ towards $4\nu q^{2p}$ will result in a hyperbolic blow-up of the left-hand-side of the hypothesis of Eq.~\eqref{eq:SuffCondOne}, thus yielding an even more rapid violation of the hypothesis. 

 Now, let us consider the second dissipation term configuration where we eliminate extrapolated Ekman dissipation but retain differential small-scale dissipation. This corresponds to choosing $\mu=0$ and $\gl=1$ (i.e. the Ekman term is now at the lower layer and not extrapolated into the surface layer), with $D_1(q)=\nu q^{2p+2}$, and $d(q)=\nu_E q^2$ (since $\gl=1$), and $D_2(q)-D_1(q)=\gD\nu q^{2p+2}$.  Proposition~\ref{prop:GeneralSuffCondition} will now reduce to  the statement given by
\begin{equation}
\frac{\gD\nu q^{2p}+\nu_E}{4\nu q^{2p}} < \frac{q^2-k^2}{k_R^2}  \implies \gD (k,q) \leq  0, \label{eq:SuffCondTwo}
\end{equation}
where the hyperbolic blow-up is no longer possible. Differential small-scale dissipation however maintains its tendency towards violating the flux inequality for increasing $\gD\nu$ since the left-hand-side in the hypothesis of Eq.~\eqref{eq:SuffCondTwo} still approaches $\gD\nu/(4\nu)$ in the limit of large wavenumbers $q$, and does not vanish. Comparing Eq.~\eqref{eq:SuffCondOne} with Eq.~\eqref{eq:SuffCondTwo}, we see that the presence of $\nu_E$ in the denominator of the left-hand-side fraction of Eq.~\eqref{eq:SuffCondOne}  is due to the use of extrapolated Ekman dissipation.

It is also interesting to consider the third dissipation term configuration in which we eliminate differential small-scale dissipation but retain extrapolated Ekman dissipation. This corresponds to choosing $\mu=-1/3$ and $\gD\nu=0$, with $D_1(q)=D_2(q)=\nu q^{2p+2}$, and $d(q)=(3/2)\nu_E q^2$ (since $\gl=3/2$). The statement of Eq.~\eqref{eq:SuffCondOne} can be simplified to read
\begin{equation}
\frac{\nu_E}{4\nu q^{2p}} < \frac{q^2-k^2}{k_R^2+(q^2-k^2)} \implies \gD (k,q) \leq  0. \label{eq:SuffCondThree}
\end{equation}
Now, let us compare Eq.~\eqref{eq:SuffCondThree} against the fourth dissipation term configuration  where both differential small-scale dissipation and extrapolated Ekman dissipation are eliminated (i.e. $\mu=0$ and $\gD\nu=0$, with  $D_1(q)=D_2(q)=\nu q^{2p+2}$, and $d(q)=\nu_E q^2$ (since $\gl=1$)). The corresponding sufficient condition  is given by 
\begin{equation}
\frac{\nu_E}{4\nu q^{2p}} \leq \frac{q^2-k^2}{k_R^2} \implies \gD(k,q) \leq 0. \label{eq:OrigSuffCond}
\end{equation}

We see that  in the absense of both extrapolated Ekman dissipation and differential small-scale dissipation,  the sufficient condition to satisfy the flux inequality is easily satisfied since the left-hand-side of Eq.~\eqref{eq:OrigSuffCond} vanishes with increasing wavenumber $q$ whereas the right-hand side increases quadratically with $q$. The only way to frustrate the sufficient condition and hope to be able to violate the flux inequality is by adjusting the hyperdissipation coefficient $\nu$ with increasing numerical resolution, such that $\nu_E k_R^2/(4\nu q_{\text{max}}^{2p+2})$ remains constant, with $q_{\text{max}}$ the maximum resolved wavenumber. Such an  adjustment of hyperdissipation was indeed necessary in the Tung-Orlando simulation \cite{article:Orlando:2003} of the two-layer quasi-geostrophic  model, opening it to criticism \cite{article:Smith:2004,article:Tung:2004}. On the other hand, the sufficient conditions for the other three cases indicate that the need for this kind of adjustment may be diminished. Extrapolated Ekman dissipation alone stabilizes the growth of the right-hand side of the sufficient condition in Eq.~\eqref{eq:SuffCondThree}  but does not stop the left-hand side from vanishing. This situation is considerably improved, as can be seen from Eq.~\eqref{eq:SuffCondOne} and Eq.~\eqref{eq:SuffCondTwo}, when we introduce differential small-scale dissipation. In fact, under the first configuration, corresponding to Eq.~\eqref{eq:SuffCondOne}, all it takes to violate the sufficient condition is to ensure that $\nu_E > 4\nu q^{2p}$ for all wavenumbers $q$ in the inertial and dissipation range. 

 It should be stressed that in the above discussion, the hypotheses given by Eq.~\eqref{eq:SuffCondOne}--\eqref{eq:SuffCondThree} are sufficient conditions but not necessary conditions. A violation of Eq.~\eqref{eq:GeneralSuffCondition} will ensure that the term $A_1(k,q)U_1(q)$ gives a positive contribution to $\gD(k,q)$. However, according to Lemma~\ref{prop:theLemma}, the contributions of $A_2 (k,q) U_2(q)$ and $A_3 (k,q) [2C_{12}(q) - U(q)]$ will remain negative, so the sign of $\gD (k,q)$ is dependant on which term gets to be dominant. Therefore, it is far from a foregone conclusion that a violation of the flux inequality is possible under the dissipation configurations considered above. However, the significant tightening of the sufficient condition with the introduction of extrapolated Ekman dissipation and differential small-scale dissipation indicates that a violation of the flux inequality may be  becoming easier to achieve, under these configurations.

\subsection{Sufficient conditions in terms of streamfunction spectra}

We would now like to consider statements providing sufficient conditions for satisfying the flux inequality, formulated in terms of the streamfunction spectra $U_1(q)$, $U_2(q)$, and $C_{12}(q)$, for the dissipation configuration given by Eq.~\eqref{eq:DissOne} and Eq.~\eqref{eq:DissTwo}, i.e. streamfunction dissipation with differential small-scale dissipation and extrapolated Ekman dissipation. These conditions constrain the spectrum $C_{12}(q)$ with respect to $U_1(q)$ and $U_2(q)$, and they imply corresponding constraints on the distribution of energy and potential enstrophy between layers, to be explored in future work. Furthermore, they are independent of the detailed definitions of the dissipation term operator spectra given by $D_1 (k)$, $D_2 (k)$, and $d(k)$. 

We will derive propositions for three separate cases. Proposition~\ref{prop:StreamfunctionAsymmetricDissipationTwo} corresponds to streamfunction dissipation without differential small-scale dissipation and without extrapolated Ekman dissipation. Proposition~\ref{prop:DiffEkmanSpectraCond} corresponds to streamfunction dissipation with both differential small-scale dissipation and extrapolated Ekman dissipation. Finally, Proposition~\ref{prop:EkmanSpectraCond} corresponds to streamfunction dissipation with extrapolated Ekman dissipation but without differential small-scale dissipation. We will see that the corresponding constraints on the streamfunction spectrum $C_{12}(q)$ become tighter upon introducing differential small-scale dissipation, extrapolated Ekman dissipation, or a combination of both.

The first step towards deriving the propositions below is to rewrite $\gD(k,q)$ in terms of $D_1(q)$, $D_2(q)$, and $d(q)$ as follows
\begin{equation}
\gD (k,q) = B_1(k,q) D_1(q) + B_2(k,q) D_2(q) + B_3 (k,q)d(q),
\end{equation}
with $B_1(k,q)$, $B_2(k,q)$, and $B_3(k,q)$ given by 
\begin{align}
B_1(k,q)  &= [2k^2-2a(q)-b(q)] U_1(q)- b(q)U_2(q)-b(q)[2C_{12}(q)-U(q)] \\
&= 2[k^2-a(q)]U_1(q)-2b(q)C_{12}(q), \\
B_2 (k,q) &= -b(q)U_1(q)+[2k^2-2a(q)-b(q)]U_2(q)-b(q)[2C_{12}(q)-U(q)] \\
&= 2[k^2-a(q)]U_2(q)-2b(q)C_{12}(q), \\
B_3 (k,q) &= -b(q)U_1(q)+[k^2-a(q)-2b(q)]\mu U_1(q)+[2k^2-2a(q)-b(q)]U_2(q) \nonumber \\ &\quad +\mu [k^2-a(q)]U_2(q)-b(q)[2C_{12}(q)-U(q)]+\mu [k^2-a(q)][2C_{12}(q)-U(q)] \\
&= 2[k^2-a(q)]U_2(q)-2b(q)C_{12}(q) \nonumber \\ &\quad +\mu [k^2-a(q)][U_1(q)+U_2(q)+2C_{12}(q)-U(q)]-2\mu b(q)U_1(q) \\
&= 2[k^2-a(q)]U_2(q)-2b(q)C_{12}(q)+\mu [k^2-a(q)]2C_{12}(q)-2\mu b(q)U_1(q). 
\end{align}

We now use the above equations to derive the following propositions:

\begin{proposition}
\label{prop:StreamfunctionAsymmetricDissipationTwo}
Assume  that $k^2-a(q)-b(q)<0$ and $b(q)<0$. We also assume  the dissipation configuration given by Eq.~\eqref{eq:GeneralDissipationMatrix} with $\mu=0$, and $d(q) \geq 0$, and $D_1(q)=D_2(q)\equiv D(q) \geq 0$ (i.e. symmetric small-scale streamfunction dissipation with a standard Ekman term). It follows that if $C_{12}(q)\leq U_2(q)$, then $\gD(k,q)\leq 0$. 
\end{proposition}

\begin{proof}
We write $\gD(k,q)$, under the assumption of symmetric small-scale dissipation (i.e. $D_1(q)=D_2(q)\equiv D(k)$), as
\begin{equation}
\gD (k,q) = [B_1(k,q)+B_2(k,q)]D(q)+B_3(k,q)d(q).
\end{equation}
 We note that from the given assumptions, we have
\begin{align}
B_1(k,q)+B_2(k,q) &= 2[k^2-a(q)]U_1(q)-2b(q)C_{12}(q)+2[k^2-a(q)]U_2(q)-2b(q)C_{12}(q) \\
&= 2[k^2-a(q)]U(q)-4b(q)C_{12}(q) \\
&= 2[k^2-a(q)-b(q)]U(q)-2b(q)[2C_{12}(q)-U(q)] \leq 0,
\end{align}
using $k^2-a(q)-b(q)<0$, $b(q)<0$, and $2C_{12}(q)-U(q) \leq 0$. From the hypothesis $C_{12}(q)\leq U_2(q)$, we can also show that
\begin{align}
B_3 (k,q) &= 2[k^2-a(q)] U_2(q)-2b(q)C_{12}(q) \\
&\leq 2[k^2-a(q)] U_2(q)-2b(q)U_2(q) = 2[k^2-a(q)-b(q)]U_2(q) \leq 0.
\end{align}
Since $d(q) \geq 0$ and $D(q) \geq 0$, it follows that $\gD(k,q)\leq 0$.
\end{proof}



\begin{proposition}
Assume that $b(q)<0$ and $k^2-a(q)-b(q)<0$. We also assume the most general dissipation configuration given by Eq.~\eqref{eq:GeneralDissipationMatrix}  with $D_1 (q)\geq 0$, and $D_2 (q)\geq 0$, and $d(q)\geq 0$, and $-1<\mu <0$ (i.e. streamfunction dissipation with both differential small-scale dissipation and extrapolated Ekman dissipation). It follows that 
\begin{enumerate}
\item If $C_{12}(q)\leq 0$, then $\gD (k,q) \leq 0$.
\item If $C_{12}(q) \leq \min \{U_1(q), U_2(q)\}$ and $U_1(q)+\mu U_2(q) \geq 0$, then $\gD(k,q) \leq 0$.
\end{enumerate}
\label{prop:DiffEkmanSpectraCond}
\end{proposition}

\begin{proof}
To show (a) we first note that $k^2-a(q) = [k^2-a(q)-b(q)]+b(q)< k^2-a(q)-b(q)<0$. Combined with the given assumptions, we find that $B_1(k,q)$, $B_2(k,q)$ and $B_3(k,q)$ satisfy 
\begin{align}
B_1 (k,q) &= 2[k^2-a(q)]U_1(q)-2b(q)C_{12}(q) \leq 2[k^2-a(q)]U_1(q) \leq 0, \\
B_2(k,q) &= 2[k^2-a(q)]U_2(q)-2b(q)C_{12}(q) \leq 2[k^2-a(q)]U_2(q) \leq 0, \\
B_3(k,q) &= 2[k^2-a(q)]U_2(q)-2b(q)C_{12}(q)+\mu [k^2-a(q)]2C_{12}(q)-2\mu b(q)U_1(q) \\
&\leq 2[k^2-a(q)]U_2(q)-2\mu  b(q)U_1(q) \leq   0.
\end{align}
Here we used the inequalities $-2b(q)C_{12}(q)\leq 0$, and $\mu [k^2-a(q)]2C_{12}(q)\leq 0$, and $2\mu b(q)U_1(q)\leq 0$, that follow from the given assumptions. It follows that $\gD(k,q)\leq 0$. 

To show (b) we use the given assumptions to show that
\begin{align}
B_1 (k,q) &= 2[k^2-a(q)]U_1(q)-2b(q)C_{12}(q) \leq  2[k^2-a(q)]U_1(q)-2b(q)U_1(q) \\
&\leq  2[k^2-a(q)-b(q)]U_1(q) \leq 0, \\
B_2(k,q) &= 2[k^2-a(q)]U_2(q)-2b(q)C_{12}(q) \leq  2[k^2-a(q)]U_2(q)-2b(q)U_2(q) \\
&= 2[k^2-a(q)-b(q)]U_2(q) \leq 0.
\end{align}
The above two inequalities for $B_1(k,q)$ and $B_2(k,q)$ are  based on the assumptions $C_{12}(q)\leq U_1(q)$ and $C_{12}(q) \leq  U_2(q)$. We also show that $B_3(k,q)$ is bounded by
\begin{align}
B_3(k,q) &=  2[k^2-a(q)]U_2(q)-2b(q)C_{12}(q)+\mu [k^2-a(q)]2C_{12}(q)-2\mu b(q)U_1(q) \\
&\leq 2[k^2-a(q)]U_2(q)-2b(q)U_2(q)+\mu [k^2-a(q)]2U_1(q)-2\mu b(q)U_1(q) \\
&= 2[k^2-a(q)-b(q)]U_2(q) + 2\mu [k^2-a(q)-b(q)]U_1(q) \\
&= 2[k^2-a(q)-b(q)][U_2(q)+\mu U_1(q)].
\end{align}
Since $k^2-a(q)-b(q)<0$ and by hypothesis $U_1(q)+\mu U_2(q) \geq 0$ it follows that $B_3(k,q) \leq 0$, and consequently  $\gD(k,q)\leq 0$.
\end{proof}

\begin{proposition}
Assume that $k^2-a(q)-b(q)<0$ and $b(q)<0$. We also assume the dissipation configuration given by Eq.~\eqref{eq:GeneralDissipationMatrix} with $-1<\mu<0$ and $D_1(q)=D_2(q) \geq 0$, and  $d(q)\geq 0$ (i.e.streamfunction dissipation with extrapolated Ekman dissipation with  and symmetric small-scale dissipation). It follows that if $C_{12}(q) \leq  \min\{U_1(q),U_2(q)\}$ then $\gD(k,q) \leq 0$.
\label{prop:EkmanSpectraCond}
\end{proposition}

\begin{proof}
Under the assumption of symmetric small-scale dissipation (i.e. $D_1(q)=D_2(q)$), we may rewrite $\gD(k,q)$ as
\begin{equation}
\gD (k,q) = [B_1(k,q)+B_2(k,q)]D(q)+B_3(k,q)d(q).
\end{equation}
 We note that from the given assumptions, we have
\begin{align}
B_1(k,q)+B_2(k,q) &= 2[k^2-a(q)]U_1(q)-2b(q)C_{12}(q)+2[k^2-a(q)]U_2(q)-2b(q)C_{12}(q) \\
&= 2[k^2-a(q)]U(q)-4b(q)C_{12}(q) \\
&= 2[k^2-a(q)-b(q)]U(q)-2b(q)[2C_{12}(q)-U(q)] \leq 0,
\end{align}
using $k^2-a(q)-b(q)<0$, $b(q)<0$, and $2C_{12}(q)-U(q) \leq 0$. We also have 
\begin{align}
B_3 (k,q)&= 2[k^2-a(q)]U_2(q)-2b(q)C_{12}(q)+\mu [k^2-a(q)]2C_{12}(q)-2\mu b(q)U_1(q) \\
&\leq 2[k^2-a(q)]U_2(q)-2b(q)U_2(q)+\mu [k^2-a(q)]2C_{12}(q)-2\mu b(q)C_{12}(q) \\
&= 2[k^2-a(q)-b(q)]U_2(q)+2\mu [k^2-a(q)-b(q)]C_{12}(q) \\
&\leq 2[k^2-a(q)-b(q)]U_2(q)+2\mu [k^2-a(q)-b(q)]U_2(q) \\
&= 2(1+\mu)[k^2-a(q)-b(q)]U_2(q) \leq 0.
\end{align}
Here, on the first line we used the assumptions $C_{12}(q) \leq U_1(q)$ and $C_{12}(q)\leq U_2(q)$ to argue that $-2b(q)C_{12}(q) \leq -2b(q)U_2(q)$ and $-2\mu b(q)U_1(q) \leq -2\mu b(q)C_{12}(q)$. The remainder of the argument continues to apply the given assumptions and it is easy to follow. We conclude that $\gD(k,q) \leq 0$.
\end{proof}

Proposition~\ref{prop:StreamfunctionAsymmetricDissipationTwo} shows that under symmetric small-scale streamfunction dissipation alone, using standard as opposed to extrapolated Ekman dissipation, the inequality $C_{12}(q)\leq U_2(q)$ implies $\gD (k,q)\leq 0$ for all wavenumbers $k<q$. We already  know that   $C_{12}(q)$ is mathematically restricted via the arithmetic-geometric mean  inequality $2|C_{12}(q)| \leq U_1(q)+U_2(q)$ over an interval of values intersecting with the constraint $C_{12}(q)\leq U_2(q)$, so the actual constraint on $C_{12}(q)$ is tighter. From Proposition~\ref{prop:DiffEkmanSpectraCond} and Proposition~\ref{prop:EkmanSpectraCond} we see that including either extrapolated Ekman dissipation or differential small-scale dissipation on top of streamfunction dissipation makes the sufficient conditions more restrictive. This is, of course, expected and consistent with the preceding discussion of the consequences of Proposition~\ref{prop:GeneralSuffCondition}. In particular, Proposition~\ref{prop:DiffEkmanSpectraCond} shows that for the more general dissipation term configuration of streamfunction dissipation with  both differential small-scale dissipation and extrapolated Ekman dissipation, if the streamfunction spectrum $C_{12}(q)$ is negative for all wavenumbers $q>k$, then the flux inequality is satisfied at  wavenumber $k$. It also shows that the restriction on the streamfunction spectrum $C_{12}(q)$ can be stretched as far as the wider inequality $C_{12}(q) \leq  \min\{U_1(q),U_2(q)\}$ if we choose to introduce the restriction $U_1(q)+\mu U_2(q)\geq 0$ on the streamfunction spectra $U_1(q)$ and $U_2(q)$. In proposition~\ref{prop:EkmanSpectraCond} we eliminate differential small-scale dissipation but retain extrapolated Ekman dissipation. This eliminates the restriction $U_1(q)+\mu U_2(q)\geq 0$ whereas   the restriction on the streamfunction spectrum $C_{12}(q)$ remains the same as in Proposition~\ref{prop:DiffEkmanSpectraCond}. This shows that the restriction $U_1(q)+\mu U_2(q)\geq 0$ originates from differential small-scale dissipation, and since $\mu$ is negative, it constitutes a non-trivial constraint on the streamfunction spectra $U_1(q)$ and $U_2(q)$. As a result, the sufficient conditions of Proposition~\ref{prop:EkmanSpectraCond} are rigorously wider than the sufficient conditions of Proposition~\ref{prop:DiffEkmanSpectraCond}.  It goes without saying that eliminating both differential small-scale dissipation and extrapolated Ekman dissipation reverts us back to Proposition~\ref{prop:StreamfunctionAsymmetricDissipationTwo}  where the stated  sufficient condition is clearly wider than that of Proposition~\ref{prop:EkmanSpectraCond}. Specifically, for $\mu=0$, the inequality $U_1(q)+\mu U_2(q)\geq 0$ reduces to the trivial inequality $U_1(q)\geq 0$. Furthermore, in the proof of proposition~\ref{prop:EkmanSpectraCond}, for $\mu=0$, we no longer need the constraint $C_{12}(q) \leq U_1(q)$ to show that $B_3(k,q)\leq 0$, and only the constraint $C_{12}(q) \leq U_2(q)$ is needed by the remainder of the proof. 

\section{Conclusions and Discussion}

We have derived rigorous sufficient conditions for satisfying the flux inequality $k^2 \Pi_E(k)-\Pi_G(k) \leq 0$ for a general $n$-layer quasi-geostrophic model with constant layer-by-layer thickness, under symmetric streamfunction dissipation. By symmetric streamfunction dissipation we mean that for every layer the dissipation term is given by the same linear Fourier-diagonal operator, applied only on the streamfunction field of the same layer.  It follows that under symmetric configurations of the dissipation terms, $n$-layer quasi-geostrophic models will indeed have a phenomenology similar to two-dimensional Navier-Stokes turbulence. Asymmetric dissipation configurations, where different dissipation operators are used on different layers, have been considered for the special case of a two-layer quasi-geostrophic model, dissipated with general streamfunction dissipation with or without extrapolated Ekman dissipation and differential small-scale dissipation. We have demonstrated that if the degree of asymmetry in the dissipation terms between the two layers is bounded as described by Proposition~\ref{prop:GeneralSuffCondition}, then the flux inequality will continue to be satisfied. Our results on the non-trivial dependence of the dissipation rate spectra of energy and potential enstrophy on the energy and potential enstrophy spectra via the streamfunction spectra, are also very relevant to the correct formulation of closure models for multi-layer quasi-geostrophic systems. 


One limitation of the current investigation is that we have disregarded the beta term, mainly to avoid the mathematical difficulties associated with the  anisotropic nature of the term.  This elimination can be tolerated, from a physical standpoint, as long as the beta term is active only in the forcing range and the baroclinic forcing at the same forcing range is powerful enough to overshadow the beta term. As long as the effect of the beta term  remains limited to large scales (i.e. planetary and synoptic scales), it will not contribute to the integrals of Eq.~\eqref{eq:EnergyFluxIntegral} and Eq.~\eqref{eq:PotentialEnstrophyFluxIntegral} and the results reported in this paper will remain entirely unaffected. The only other physical assumption inherent in these results is that forcing via the baroclinic instability is limited to large scales.  The propositions 1-5 are mathematically rigorous and do not require these assumptions, however the assumptions come into play at the very last step   where the conclusion of propositions 1-5 (i.e. $\gD (k,q) \leq  0$) is used to infer the flux inequality itself. On the other hand, no phenomenological assumptions about any spectrum are needed at any step of the argument.


For the case of the two-layer quasi-geostrophic model, we have seen that, starting from a streamfunction dissipation configuration,  adding either extrapolated Ekman dissipation or differential small-scale dissipation (or both) tends to tighten the sufficient conditions for satisfying the flux inequality. This suggests   that the flux inequality $k^2\Pi_E (k)-\Pi_G(k)\leq 0$ may be more easily violated under these more general dissipation configurations.  A violation of the flux inequality beyond a wavenumber $k_t$ would then allow a downscale energy flux large enough to result in a transition from $k^{-3}$ to $k^{-5/3}$ scaling in the energy spectrum near the wavenumber $k_t$ \cite{article:Tung:2005,article:Tung:2005:1,article:Tung:2007}. However, while there is a plausible physical motivation for using extrapolated Ekman dissipation, there is no obvious physical motivation for introducing an asymmetric configuration of the small-scale dissipation terms. We would therefore like to expand on the reasons why we believe that this is an idea worth pursuing.


Any kind of small-scale dissipation in quasi-geostrophic models is not physical but  is tolerated mainly  because it is intended to model the dissipative mechanisms that exist at smaller scales where quasi-geostrophic dynamics breaks down and three-dimensional dynamics becomes dominant. Lindborg \cite{article:Lindborg:2009} estimates that quasi-geostrophic dynamics break down at a length scale of about $100$km. However,  the scaling transition wavenumber $k_t$ of the Nastrom-Gage spectrum \cite{article:Gage:1979,article:Nastrom:1986,article:Jasperson:1984,article:Gage:1984}, and consequently the breakdown of the flux inequality, occurs at a greater length scale of about $1000$km to $700$km in wavelength, which is still within the quasi-geostrophic regime. The hypothesis underlying the quasi-geostrophic modeling of atmospheric turbulence is that the locality of the coexisting downscale potential enstrophy cascade and downscale energy cascade shields them from the three-dimensional dominated regime at the smallest scales. Both cascades are furthermore protected by the continuing conservation of potential enstrophy under the stratified turbulence dynamics that becomes dominant at scales less than $100$km.  The above considerations suggest the hypothesis that three-dimensional effects will not contaminate the nonlinear quasi-geostrophic dynamics driving the coexisting cascades of potential enstrophy and energy in the quasi-geostrophic regime, which allows us to model small-scale three-dimensional  processes, as seen from the quasi-geostrophic regime's point of view,  via small-scale hyperdiffusion terms applied to all layers. That said, there is the non-local effect that the anomalous energy sink, provided by the three-dimensional  regime, can inflict on the quasi-geostrophic regime, and that is to boost  the downscale energy dissipation rate, thereby increasing the downscale energy flux passing through the quasi-geostrophic regime of length scales and  moving  the transition wavenumber $k_t$ deep into the inertial range. This is why we propose that if differential small-scale dissipation can be shown to achieve an equivalent effect, then it should be accepted as a more realistic configuration, for modelling purposes. As we have mentioned in the introduction, the main weakness of multi-layer quasi-geostrophic models is that they disregard the surface quasi-geostrophic dynamics at the lowest layer.  Perhaps, asymmetric small-scale dissipation can be thought of as a crude way to compensate for the absence of surface quasi-geostrophic dynamics. 

\section*{Acknowledgement}

The idea of a flux inequality was first brought up in email communication between Sergey Danilov with the author and Ka-Kit Tung, in the context of two-dimensional Navier-Stokes turbulence.

%
%

\appendix

\section{Proof that $E(k)$ is always positive}
\label{app:EkPositiveDefinite}

In this appendix we  show that if the matrix $L_{\ab}(k)$  satisfies the following conditions:
\begin{align}
& L_{\ab}(k)\geq 0, \text{ for } \ga\neq\gb, \label{eq:PositiveDefiniteOne} \\ 
& \sum_{\gb} L_{\ab}(k) \leq 0, \label{eq:PositiveDefiniteTwo}
\end{align}
then the energy spectrum $E(k)$ will be always positive with $E(k)\geq 0$.

We begin by rewriting Eq.~\eqref{eq:Ek} as follows: 
\begin{align}
E(k) &= -\sum_{\ab} L_{\ab}(k)C_{\ab}(k) = -\sum_{\ga} L_{\ga\ga}(k)U_\ga (k) -\sum_{\substack{\ab \\ \ga\neq\gb}} L_{\ab}(k)C_{\ab}(k) \\
&= -\sum_{\ga} L_{\ga\ga}(k)U_\ga (k) - \frac{1}{2}\sum_{\substack{\ab \\ \ga\neq\gb}} L_{\ab}(k) [U_{\ga}(k)+U_{\gb}(k)] - \frac{1}{2}\sum_{\substack{\ab \\ \ga\neq\gb}} L_{\ab}(k) [2C_{\ab}(k)-U_{\ga}(k)-U_{\gb}(k)] \label{eq:EkPositiveFirstStep} \\
&= -\sum_{\ga} L_{\ga\ga}(k)U_\ga (k) -\frac{1}{2}\sum_{\substack{\ab \\ \ga\neq\gb}} [L_{\ab}(k)+L_{\ba}(k)]U_\ga (k)  - \frac{1}{2}\sum_{\substack{\ab \\ \ga\neq\gb}} L_{\ab}(k) [2C_{\ab}(k)-U_{\ga}(k)-U_{\gb}(k)] \label{eq:EkPositiveSecondStep} \\
&= -\sum_{\ga}\biggl[  L_{\ga\ga}(k)+ \frac{1}{2}\sum_{\substack{\gb \\ \ga\neq\gb}}[L_{\ab}(k)+L_{\ba}(k)]\biggr] U_\ga (k)   - \frac{1}{2}\sum_{\substack{\ab \\ \ga\neq\gb}} L_{\ab}(k) [2C_{\ab}(k)-U_{\ga}(k)-U_{\gb}(k)] \\
&= -\sum_{\ga}\biggl[  \sum_{\gb} L_{\ab}(k) \biggr] U_\ga (k)  - \frac{1}{2}\sum_{\substack{\ab \\ \ga\neq\gb}} L_{\ab}(k) [2C_{\ab}(k)-U_{\ga}(k)-U_{\gb}(k)].
\end{align}
The assumption $L_{\ab}(k)=L_{\ba}(k)$ is used in the key step between Eq.~\eqref{eq:EkPositiveFirstStep} and Eq.~\eqref{eq:EkPositiveSecondStep}. We note that $U_\ga (k)\geq 0$, since $U_{\ga}(k)$ is always positive, and from the arithmetic-geometric mean inequality, $2C_{\ab}(k)-U_\ga (k)-U_\gb (k)\leq 0$. Combining these with the assumptions given by Eq.~\eqref{eq:PositiveDefiniteOne} and Eq.~\eqref{eq:PositiveDefiniteTwo}, we see that both terms in our expression for $E(k)$ are positive and therefore $E(k)\geq 0$.


\section{Derivation of dissipation rate spectra}
\label{app:DissipationRateSpectra}

In this appendix, we will show that the energy dissipation rate spectrum $D_E (k)$ and the layer-by-layer potential enstrophy dissipation rate spectra $D_{G_\ga} (k)$ are given by
\begin{align}
D_E (k) &= 2 \sum_{\ab} D_{\ab} (k) C_{\ab} (k), \\ 
D_{G_\ga} (k) &= -2 \sum_{\bc} L_{\ab}(k) D_{\ac} (k) C_{\bc} (k).
\end{align}
The proof mirrors the argument used in Ref.~\cite{article:Gkioulekas:p15} to derive the energy forcing spectrum and the potential enstrophy forcing spectrum for the same model. We begin by writing the governing equation for the streamfunction field $\gy_\ga$ as
\begin{equation}
\pderiv{\psi_\ga}{t}+ \sum_\gb \ccL_{\ab}^{-1} J(\psi_\gb, q_\gb) = \sum_{\bc} \ccL_{\ab}^{-1} \ccD_{\bc} \psi_\gc + \sum_\gb \ccL_{\ab}^{-1} f_\gb. \label{eq:GovEqStreamFunc}
\end{equation}
Differentiating the streamfunction spectrum $C_{\ab} (k)$ with respect to time gives
\begin{equation}
\pderiv{C_{\ab}(k)}{t} = \innerf{\pderiv{\psi_\ga}{t}}{\psi_\gb}{k} + \innerf{\psi_\ga}{\pderiv{\psi_\gb}{t}}{k}, \label{eq:TimeDerCab}
\end{equation}
and we may write a governing equation for $C_{\ab} (k)$ in the form:
\begin{equation}
\pderiv{C_{\ab}(k)}{t} + \cT_{\ab}(k) = -\cD_{\ab}(k) + \cF_{\ab}(k). \label{eq:GovEqCab}
\end{equation}
Here, $\cT_{\ab}(k)$ is the contribution from the nonlinear Jacobian term, $\cD_{\ab}(k)$ is the contribution from the dissipation term, and $\cF_{\ab}(k)$ is the contribution from the forcing term. The dissipation term $\cD_{\ab} (k)$ can now be obtained by replacing in Eq.~\eqref{eq:TimeDerCab} the streamfunction time-derivative $\pderivin{\psi_\ga}{t}$ with the dissipation term $\sum_{\bc} \ccL_{\ab}^{-1} \ccD_{\bc} \psi_\gc$. This gives
\begin{align}
\cD_{\ab} (k) &= -\innerf{\sum_{\cd} \ccL_{\ac}^{-1} \ccD_{\cd} \psi_\gd}{\psi_\gb}{k} - \innerf{\psi_\ga}{\sum_{\cd} \ccL_{\bc}^{-1} \ccD_{\cd} \psi_\gd}{k}\\ 
&= -\sum_{\cd} [L_{\ac}^{-1}(k) D_{\cd}(k) C_{\bd} (k) + L_{\bc}^{-1}(k) D_{\cd}(k) C_{\ad} (k) ].
\end{align}
We may now easily write the dissipation rate spectra $D_E (k)$ and $D_G(k)$ by applying on $\cD_{\ab} (k)$ the linear operators indicated by Eq.~\eqref{eq:Ek} and Eq.~\eqref{eq:Gk}. We therefore find that the energy dissipation rate energy spectrum $D_E (k)$ is given by
\begin{align}
D_E (k) &= -\sum_{\ab} L_{\ab}(k) \cD_{\ab}(k) = \sum_{\abcd} [L_{\ab}(k) L_{\ac}^{-1}(k) D_{\cd}(k) C_{\bd} (k) + L_{\ab}(k) L_{\bc}^{-1}(k) D_{\cd}(k) C_{\ad} (k)] \\
&= \sum_{\bcd} \biggl[\sum_\ga L_{\ba}(k) L_{\ac}^{-1}(k)\biggr] D_{\cd}(k) C_{\bd} (k) + \sum_{\acd} \biggl[\sum_\gb L_{\ab}(k) L_{\bc}^{-1}(k)\biggr] D_{\cd}(k) C_{\ad} (k) \\
&= \sum_{\bcd} \gd_{\bc} D_{\cd}(k) C_{\bd} (k) + \sum_{\acd} \gd_{\ac} D_{\cd}(k) C_{\ad} (k) \\ 
&= \sum_{\bd} D_{\bd}(k) C_{\bd} (k) + \sum_{\cd} D_{\cd}(k) C_{\cd} (k) = 2 \sum_{\ab} D_{\ab} (k) C_{\ab} (k).
\end{align}
The layer-by-layer potential enstrophy spectrum $D_{G_\ga}(k)$ is likewise given by
\begin{align}
D_{G_\ga}(k) &= \sum_{\bc} L_{\ab} (k) L_{\ac}(k) \cD_{\bc}(k) = -\sum_{\bcde} L_{\ab} (k) L_{\ac}(k) [L_{\bd}^{-1}(k) D_{\de}(k) C_{\ce} (k) + L_{\cd}^{-1}(k) D_{\de}(k) C_{\be} (k) ] \\ 
&= -\sum_{\cde} \gd_{\ad} L_{\ac}(k) D_{\de}(k) C_{\ce} (k) - \sum_{\bde} \gd_{\ad} L_{\ab} (k) D_{\de}(k) C_{\be} (k) \\ 
&= -\sum_{\ce} L_{\ac}(k) D_{\ga\gee}(k) C_{\ce} (k) - \sum_{\be} L_{\ab} (k) D_{\ga\gee}(k) C_{\be} (k) \\ 
&= -2 \sum_{\bc} L_{\ab}(k) D_{\ac} (k) C_{\bc} (k).
\end{align}
The corresponding conservation laws read
\begin{align}
&\pderiv{E(k)}{t} + \pderiv{\Pi_E(k)}{t} = -D_E (k) + F_E (k),\\
&\pderiv{G(k)}{t} + \pderiv{\Pi_G(k)}{t} = -D_G (k) + F_G (k).
\end{align}
We see that positive $D_E(k)$ and $D_G(k)$ correspond to the case where the dissipation terms are truly dissipative. This concludes the argument.

%
%

\bibliographystyle{elsarticle-num}
\bibliography{references}
\end{document}